\documentclass[10pt,twocolumn,letterpaper]{article}
\usepackage{cvpr}
\usepackage{amsthm}
\usepackage{amsmath,amssymb} 
\usepackage[T1]{fontenc}
\usepackage{etex}
\usepackage{graphicx}
\usepackage{amsxtra}
\usepackage{amsfonts}
\usepackage{nicefrac}
\usepackage{microtype}
\usepackage{url}
\usepackage{soul}

\usepackage{color}
\usepackage{bbm}
\usepackage{pgfplots}
\usepackage{tikz}

\usepackage{wrapfig}
\usepackage{sidecap}
\usepackage[ruled]{algorithm2e}

\usepackage{colortbl}
\newcolumntype{g}{>{\columncolor[gray]{0.85}}c}

\usepackage{epsfig}
\usepackage{stmaryrd} 
\usepackage{mathtools}

\usepackage{multirow}
\usepackage{fixmath}
\usepackage{booktabs}
\usepackage{pbox}
\usepackage{enumerate}
\usepackage{enumitem}

\usepackage[breaklinks=true,letterpaper=true,colorlinks,bookmarks=false]{hyperref}

\newcommand{\BV}{\mathbb{V}}
\newcommand{\BE}{\mathbb{E}}
\newcommand{\BG}{\mathbb{G}}

\newcommand{\FV}{\mathsf{V}}
\newcommand{\FE}{\mathsf{E}}

\newcommand{\IPSLP}{IRPS-LP}

\newcommand{\R}{\mathbb{R}}

\newcommand{\la}{\langle}
\newcommand{\ra}{\rangle}
\newcommand{\T}{\top}

\newcommand{\dcup}{\dot{\cup}}

\DeclareMathOperator*{\conv}{conv}

\def\clap#1{\hbox to 0pt{\hss#1\hss}}

\def\mathclap{\mathpalette\mathclapinternal}

\def\mathclapinternal#1#2{%
\clap{$\mathsurround=0pt#1{#2}$}}

\providecommand{\abs}[1]{\lvert#1\rvert}

\definecolor{fettrot}{RGB}{255,10,10}

\newtheorem{definition}{Definition}

\newtheorem{proposition}{Proposition}

\usepackage{times}
\usepackage{placeins}
\usepackage{longtable}

	\definecolor{mycolor1}{HTML}{009900}
    \definecolor{mycolor2}{HTML}{FF0000}
    \definecolor{green}{HTML}{00FF00}
    \definecolor{red}{HTML}{FF0000}
    \usetikzlibrary{decorations.pathmorphing}
  \tikzstyle{cut-edge}=[red]
    \tikzstyle{vertex}=[circle, draw, fill=white, inner sep=0pt, minimum width=1ex]
    \tikzset{every picture/.append style={baseline,scale=1.1}}

\cvprfinalcopy 

\newcommand{\bjoern}[2]{#2}

\setitemize{topsep=2.5pt,parsep=1.5pt,partopsep=1.5pt,leftmargin=0.4cm}

\begin{document}

\title{A Message Passing Algorithm for the Minimum Cost Multicut Problem}
\author{Paul Swoboda and Bj\"orn Andres}
  
\maketitle

\begin{abstract}
We propose a dual decomposition and linear program relaxation of the \textsc{np}-hard minimum cost multicut problem. Unlike other polyhedral relaxations of the multicut polytope, it is amenable to efficient optimization by message passing. Like other polyhedral relaxations, it can be tightened efficiently by cutting planes. We define an algorithm that alternates between message passing and efficient separation of cycle- and odd-wheel inequalities. This algorithm is more efficient than state-of-the-art algorithms based on linear programming, including algorithms written in the framework of leading commercial software, as we show in experiments with large instances of the problem from applications in computer vision, biomedical image analysis and data mining.
\end{abstract}


\section{Introduction}
\bjoern{Segmentation of data}{Decomposing a graph} into meaningful \bjoern{subsets}{clusters} is a fundamental primitive in computer vision, \bjoern{neuroscience}{biomedical image analysis} and data mining.
\bjoern{In the unsupervised setting we have no information which the classes are we shall assign the data to, the only information being similarity/dissimilarity of a subset of object pairs.
A natural formulation of this task is the multicut problem~\cite{ChopraMulticut}, also known as correlation clustering.
Individual objects are nodes of a graph, with edges being given cost according to similarity/dissimilarity.
A multicut is a partitioning of the graph into an arbitrary number of components.
The cost of a multicut is the sum of weights over all edges with endpoints in different components.}
{In settings where no information is given about the number or size of clusters, and information is given only about the pairwise similarity or dissimilarity of nodes, a canonical mathematical abstraction is the minimum cost multicut (or correlation clustering) problem~\cite{ChopraMulticut}.
The feasible solutions of this problem, multicuts, relate one-to-one to the decompositions of the graph.
A multicut is the set of edges that straddle distinct clusters.
The cost of a multicut is the sum of costs attributed to its edges.}

In the field of computer vision, the minimum cost multicut problem has been applied in \cite{BreakAndConquerAlushGoldberger,ImageSegmentationClosednessAndres,PlanarCorrelationClusteringYarkony,SegmentingNonPlanarAffinityAndresYarkony} \bjoern{for}{to the task of} unsupervised image segmentation \bjoern{on}{defined by} the \bjoern{Berkeley Segmentation Dataset}{BSDS data sets and benchmarks}~\cite{BerkeleySegmentationDataset} \bjoern{and state-of-the-art results have been obtained}.
In \bjoern{computational neuroscience}{the field of biomedical image analysis}, the \bjoern{}{minimum cost} multicut problem has been applied \bjoern{for}{to an image segmentation task for} connectomics~\cite{ClosedSurfaceSegmentationConnectomics}.
In \bjoern{}{the field of} data mining, applications include~\cite{DeduplicationMulticut,ClusteringQueryRefinement,ClusteringSparseGraphs,CorrelationClusteringMapReduce}.
\bjoern{Unfortunately, the multicut problem is NP-hard to solve.
Additionally, especially in connectomics the need to solve very large schale problems with millions of edges accurately has posed considerable challenges for existing solver.}
{As the minimum cost multicut problem is \textsc{np}-hard~\cite{Bansal2004,Demaine2006}, even for planar graphs~\cite{OptimalCoalitionStructureGeneration} large and complex instances with millions of edges, especially those for connectomics, pose a challenge for existing algorithms.}

\textbf{Related Work.}
Due to the importance of multicuts for applications, many algorithms for the minimum cost multicut problem have been proposed.
They are grouped below into three categories: primal feasible local search algorithms, linear programming algorithms and fusion algorithms.

\emph{Primal feasible local search algorithms}~\cite{MLforCoreferenceResolution,ImprovingMLforCoreferenceResolution,ClusteringAggregation,ConversationDisentanglementCC,BoundingAndComparingCorrelationClustering} 
attempt to improve an initial feasible solution by means of local transformations from a set that can be indexed or searched efficiently.
Local search algorithms are practical for large instances, as the cost of all operations is small compared to the cost of solving the entire problem at once.
On the downside, the feasible solution that is output typically depends on the initialization. 
And even if a solution is found, optimality is not certified, as no lower bound is computed.

\emph{Linear programming algorithms}~\cite{KappesMulticut,HigherOrderSegmentationByMulticuts,SungwoongHigherOrderCorrelationClustering,NowozinMulticut,MulticutDDYarkony} operate on an outer polyhedral relaxation of the feasible set.
Their output is independent of their initialization and provides a lower bound.
This lower bound can be used directly inside a branch-and-bound search for certified optimal solutions.
Alternatively, the LP relaxation can be tightened by cutting planes. 
Several classes of planes are known that define a facet of the multicut polytope and can be separated efficiently \cite{ChopraMulticut}.
On the downside, algorithms for general LPs that are agnostic to the structure of the multicut problem scale super-linearly with the size of the instance.

\emph{Fusion algorithms} attempt to combine feasible solutions of subproblems obtained by combinatorial or random procedures into successively better multicuts.
The fusion process can either rely on column generation~\cite{PlanarCorrelationClusteringYarkony},
binary quadratic programming~\cite{CGC} or any algorithm for solving integer LPs~\cite{FusionMoveCC}.
In particular, \cite{PlanarCorrelationClusteringYarkony} provides dual lower bounds but is restricted to planar graphs.
\cite{CGC,FusionMoveCC} explore the primal solution space in a clever way, but do not output dual information.

\textbf{Outline.}
Below, a discussion of preliminaries (Sec.~\ref{section:preliminaries}) is followed by the definition of our proposed decomposition (Sec.~\ref{sec:DD}) and algorithm (Sec.~\ref{sec:Algorithm}) for the minimum cost multicut problem.
Our approach combines the efficiency of local search with the lower bounds of LPs and the subproblems of fusion, as we show in experiments with large and diverse instances of the problem (Sec.~\ref{sec:Experiments}).
All code and data will be made publicly available upon acceptance of the paper.

\section{Preliminaries}
\label{section:preliminaries}

\subsection{Minimum Cost Multicut Problem}

A \emph{decomposition (or clustering)} of a graph $G = (V,E)$ is a partition $V_1 \dcup \ldots \cup V_k$ of the node set $V$ such that $V_i \cap V_j = \varnothing$ $\forall i\neq j$ and every cluster $V_i$, $i=1,\ldots,k$ is connected.
The \emph{multicut} induced by a decomposition is the subset of those edges that straddle distinct clusters (cf.~Fig.~\ref{fig:Multicut}).
Such edges are said to be \emph{cut}.
Every multicut induced by any decomposition of $G$ is called a multicut of $G$.
We denote by $\mathcal{M}_G$ the set of all multicuts of $G$.

Given, for every edge $e \in E$, a cost $c_e \in \mathbb{R}$ of this edge being cut, the instance of the \emph{minimum cost multicut problem} w.r.t.~these costs is the optimization problem \eqref{eq:MulticutSetFormulation} whose feasible solutions are all multicuts of $G$.
For any edge $\{v,w\} = e \in E$, negative costs $\theta_e < 0$ favour the nodes $v$ and $w$ to be in distinct components.
Positive costs $\theta_e > 0$ favour these nodes to lie in the same component.
\begin{equation}
\min_{M \in \mathcal{M}_G} \sum_{e \in M} \theta_e
    \label{eq:MulticutSetFormulation}
\end{equation}

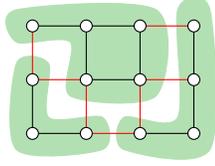
\begin{figure}
   \centering
   \begin{tikzpicture}[scale=0.65]
    \draw[draw=mycolor1!30, fill=mycolor1!30] plot[smooth cycle, tension=0.5] coordinates
        {(-0.3, 2.3) (2.3, 2.3) (2.3, 0.7) (0.7, 0.7) (0.7, 1.7) (-0.3, 1.7)};
    \draw[draw=mycolor1!30, fill=mycolor1!30] plot[smooth cycle, tension=0.5] coordinates
        {(-0.3, -0.3) (-0.3, 1.3) (0.3, 1.3) (0.3, 0.3) (1.3, 0.3) (1.3, -0.3)};
    \draw[draw=mycolor1!30, fill=mycolor1!30] plot[smooth cycle, tension=0.5] coordinates
        {(1.7, -0.3) (1.7, 0.3) (2.7, 0.3) (2.7, 2.3) (3.3, 2.3) (3.3, -0.3)};
	\draw (0, 0) -- (0, 1);
	\draw (0, 0) -- (1, 0);
	\draw (0, 2) -- (1, 2);
	\draw (1, 1) -- (1, 2);
	\draw (1, 1) -- (2, 1);
	\draw (1, 2) -- (2, 2);
	\draw (2, 1) -- (2, 2);
	\draw (2, 0) -- (3, 0);
	\draw (3, 0) -- (3, 1);
	\draw (3, 1) -- (3, 2);
	\draw[style=cut-edge] (0, 1) -- (0, 2);
	\draw[style=cut-edge] (0, 1) -- (1, 1);
	\draw[style=cut-edge] (1, 0) -- (1, 1);
	\draw[style=cut-edge] (1, 0) -- (2, 0);
	\draw[style=cut-edge] (2, 0) -- (2, 1);
	\draw[style=cut-edge] (2, 1) -- (3, 1);
	\draw[style=cut-edge] (2, 2) -- (3, 2);
	\node[style=vertex] at (0, 0) {};
	\node[style=vertex] at (1, 0) {};
	\node[style=vertex] at (0, 1) {};
	\node[style=vertex] at (0, 2) {};
	\node[style=vertex] at (1, 1) {};
	\node[style=vertex] at (1, 2) {};
	\node[style=vertex] at (2, 1) {};
	\node[style=vertex] at (2, 2) {};
	\node[style=vertex] at (2, 0) {};
	\node[style=vertex] at (3, 0) {};
	\node[style=vertex] at (3, 1) {};
	\node[style=vertex] at (3, 2) {};
   \end{tikzpicture}
   \caption{Depicted above is a decomposition of a graph into three components (\textcolor{mycolor1}{green}). 
   The multicut induced by this decomposition consists of the edges that straddle distinct components (\textcolor{red}{red}).}
   \label{fig:Multicut}
   \end{figure}


This problem is \textsc{np}-hard~\cite{Bansal2004,Demaine2006}, even for planar graphs~\cite{OptimalCoalitionStructureGeneration}.
Below, we recapitulate its formulation as a binary LP and then turn to LP relaxations:
For any 01-labeling $x \in \{0,1\}^E$ of the edges of $G$, the subset $x^{-1}(1)$ of those edges labeled 1 is a multicut of $G$ if and only if $x$ satisfies the system \eqref{eq:mp-cycles} of \emph{cycle inequalities} \cite{ChopraMulticut}.
Hence, \eqref{eq:MulticutSetFormulation} can be stated equivalently in the form of the binary LP \eqref{eq:mp-objective}--\eqref{eq:mp-integrality}.
\begin{align}
\min_{x \in \mathbb{R}^E} \quad
    & \sum_{e \in E} \theta_e x_e 
        \label{eq:mp-objective} \\
\textnormal{subject to} \quad
    & \forall C \in \textnormal{cycles}(G):
        \quad x_e \leq \sum_{\mathclap{e' \in C \setminus \{e\}}} x_{e'}
        \label{eq:mp-cycles} \\
    & x \in \{0,1\}^E 
        \label{eq:mp-integrality}
\end{align}
An LP relaxation is obtained by replacing the integrality constraints \eqref{eq:mp-integrality} by $x \in P$ with $P \subseteq [0,1]^E$.
This results in an outer relaxation of the \emph{multicut polytope}, which is the convex hull of the characteristic functions of all multicuts of $G$.
The LP relaxation obtained for $P := [0,1]^E$, i.e., with only the cycle inequalities, will not in general be tight.

A tighter LP relaxation is obtained by enforing also the \emph{odd wheel inequalities} \cite{ChopraMulticut}.
A \emph{$k$-wheel} is a cycle in $G$ with $k$ nodes all of which are connected to an additional node $u \in V$ that is not in the cycle and is called the \emph{center} of the $k$-wheel (cf.~Fig.~\ref{fig:OddWheel}).
For any odd number $k \in \mathbb{N}$, any $k$-wheel of $G$, the cycle $C = (v_1 v_2, \ldots, v_k v_1)$ and the center $u$ of the $k$-wheel, every characteristic function $x \in \{0,1\}^E$ of a multicut $x^{-1}(1)$ of $G$ satisfies the odd wheel inequality
\vspace{-1ex} 
\begin{align}
\sum_{i=1}^k x_{v_i v_{i+1}} - \sum_{i=1}^k x_{u v_i} \leq \left\lfloor \tfrac{k}{2} \right\rfloor
\quad\textnormal{with}\quad
v_{k+1} := v_1
\enspace .
\label{eq:OddWheelConstraint}
\end{align}

For completeness, we note that other inqualities known to further tighten the LP relaxation can be included in our algorithm, e.g., the \emph{bicycle inequalities}~\cite{ChopraMulticut} defind on graphs as in Fig.~\ref{fig:BicycleWheel}.
We, however, do not consider inequalities other than cycles and odd wheels in the algorithm we propose.

\tikzstyle{every node}=[circle,draw,minimum width=3ex,inner sep=0ex]
\begin{figure}
\small
\begin{minipage}{0.24\textwidth}
\centering
\begin{tikzpicture}[scale=0.6]
\draw (0,0) node (u)  {$u$};
\draw \foreach \x in {1,2,...,5} 
{
(\x*72+18:1.5) node (v\x) {$v_{\x}$}
(v\x) -- (u)
};
\draw (v1) -- (v2) -- (v3) -- (v4) -- (v5) -- (v1);
\end{tikzpicture}\\[1ex]
\caption{Odd Wheel}
\label{fig:OddWheel}
\end{minipage}%
\begin{minipage}{0.24\textwidth}
\centering
\begin{tikzpicture}[scale=0.6]
\draw (-0.4, -0.4) node (u1)  {$u_1$};
\draw (0.4, 0.4) node (u2)  {$u_2$};
\draw (u1) -- (u2);
\draw \foreach \x in {1,2,...,5} 
{
(\x*72+18:1.5) node (v\x) {$v_{\x}$} 
};
\draw (v1) -- (v2) -- (v3) -- (v4) -- (v5) -- (v1);

\draw (v1) -- (u1);
\draw (v2) -- (u1);
\draw (v3) -- (u1);
\draw (v4) -- (u1);
\draw (v5) -- (u1);

\draw (v1) -- (u2);
\draw (v2) -- (u2);
\draw (u2) to[bend left=30] (v3);
\draw (v4) -- (u2);
\draw (v5) -- (u2);
\end{tikzpicture}\\[1ex]
\caption{Odd Bicycle Wheel}
\label{fig:BicycleWheel}
\end{minipage}
\end{figure}

\subsection{Integer relaxed pairwise separable LPs}

LP relaxations of the multicut problem can in principle be solved with algorithms for general LPs which are available in excellent software such as CPlex~\cite{cplex} and Gurobi~\cite{gurobi}.
However, these algorithms scale super-linearly with the size of the problem and are hence impractical for large instances.

We define in Sec.~\ref{sec:DD} an LP relaxation of the multicut problem in form of an \IPSLP\ (Def.~\ref{def:\IPSLP}).
\IPSLP s are a special case of dual decomposition \cite{guignard1987lagrangean-TandA}.
In Def.~\ref{def:\IPSLP}, every $i \in \BV$ defines a subproblem, and every edge $ij \in \BE$ defines a dependency of subproblems.
Def.~\ref{def:\IPSLP} is more specific in that, firstly, the subproblems are binary and, secondly, the linear constraints \eqref{eq:IPSLP-constraints} that describe the dependence of subproblems are defined by 01-matrices that map 01-vectors to 01-vectors.
\IPSLP s are amenable to efficient optimization by message passing in the framework of \cite{ConvergentMessagePassingNIPS}.

\begin{definition}[IRPS-LP \cite{ConvergentMessagePassingNIPS}]
\label{def:\IPSLP}
Let $N \in \mathbb{N}$ and let $\BG = (\BV, \BE)$ be a graph with $\BV = \{1, \ldots, N\}$.
For every $j \in \BV$, let $d_j \in \mathbb{N}$, let $X_j \subseteq \{0,1\}^{d_j}$, 
  and let $\theta_j \in \mathbb{R}^{d_j}$.
  Let $\Lambda := \conv (X_1) \times \cdots \times \conv (X_N)$. 
For every $\{j,k\} = e \in \BE$, let $m_e \in \mathbb{N}$,
$A_{(j,k)} \in \{0,1\}^{m_e \times d_j}$ and $A_{(k,j)} \in \{0,1\}^{m_e \times d_k}$ such that
\begin{align}
\forall x \in X_j: \quad A_{(j,k)} x \in \{0,1\}^{m_e} \\
\forall x \in X_k: \quad A_{(k,j)} x \in \{0,1\}^{m_e}
\end{align}
Then, the LP written below is called \emph{integer relaxed pairwise separable} w.r.t.~the graph $\BG$.
\begin{align}
\min_{\mu \in \Lambda} \quad
	& \sum_{j \in \BV} \sum_{k = 1}^{d_j} \theta_{jk} \mu_{jk} \label{eq:IPSLP-objective}\\
\textnormal{subject to} \quad
	& \forall \{j,k\} \in E: \quad A_{(j,k)} \mu_j = A_{(k,j)} \mu_k \label{eq:IPSLP-constraints}
\end{align}
\end{definition}

\section{Dual Decomposition}
\label{sec:DD}

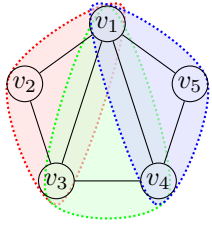
\begin{figure}
\centering
\begin{tikzpicture}[scale=0.7]
\draw (1*72+9:1.8) node[draw=none] (h1) {};
\draw (2*72+18:1.8) node[draw=none] (h2) {};
\draw (3*72+27:1.8) node[draw=none] (h3) {};
\draw[thick,red,densely dotted,fill=red!15,fill opacity=0.6] plot[smooth cycle,tension=0.6] coordinates {(h1) (h2) (h3)};
\draw (1*72+18:1.8) node[draw=none] (j1) {};
\draw (3*72+18:1.8) node[draw=none] (j2) {};
\draw (4*72+18:1.8) node[draw=none] (j3) {};
\draw[thick,green,densely dotted,fill=green!15,fill opacity=0.6] plot[smooth cycle,tension=0.6] coordinates {(j1) (j2) (j3)};
\draw (1*72+27:1.8) node[draw=none] (k1) {};
\draw (4*72+9:1.8) node[draw=none] (k2) {};
\draw (5*72+18:1.8) node[draw=none] (k3) {};
\draw[thick,blue,densely dotted,fill=blue!15,fill opacity=0.6] plot[smooth cycle,tension=0.6] coordinates {(k1) (k2) (k3)};
\draw \foreach \x in {1,2,...,5} 
{
	(\x*72+18:1.5) node (v\x) {$v_{\x}$} 
};
\draw (v1) -- (v2) -- (v3) -- (v4) -- (v5) -- (v1);
\draw (v1) -- (v3);
\draw (v1) -- (v4);
\end{tikzpicture}\\[2ex]
\caption{Depicted above is a triangulated cycle (black) covered by three triangles (\textcolor{red}{red}, \textcolor{green}{green} and \textcolor{blue}{blue})}
\label{fig:CycleTriangulation}
\end{figure}



A straight-forward decomposition of the minimum cost multicut problem \eqref{eq:mp-objective}--\eqref{eq:mp-integrality} in the form of an \IPSLP\ (Def.~\ref{def:\IPSLP}) consists of one subproblem for every edge, one subproblem for every cycle inequality and one subproblem for every odd-wheel inequality.
From a computational perspective, it is however advantageous to triangulate cycles and odd wheels, and to consider the resulting smaller subproblems.
Below, three classes of subproblems are defined rigorously.

\paragraph{Edge Subproblems.}
For every edge $e \in E$, we consider a subproblem $e \in \BV$ with the feasible set $X_e := \{0,1\}$, encoding whether edge $e$ is cut (1) or uncut (0).

\paragraph{Triangle Subproblems}
For every cycle $C = \{v_1 v_2, v_2 v_3,\ldots v_k v_1\} \subseteq E$, we consider the triangles $v_1 v_2 v_3$ to $v_{k-1} v_k v_1$, as depicted in Fig.~\ref{fig:CycleTriangulation}.
If some edge $uv$ of a triangle $C_i$ is not in $E$, we add it to $E$ with cost zero, i.e., we triangulate the cycle in $G$.
For each triangle $uvw$, we introduce a subproblem $uvw \in \BV$ whose feasible set consists of the five feasible multicuts of the triangle, i.e., $X_{uvw} := \{(0,0,0),(0,1,1),(1,0,1),(1,1,0),(1,1,1)\}$.

\paragraph{Lollipop Subproblems}
For every odd number $k \in \mathbb{N}$ and every $k$-wheel of $G$ consisting of a center node $u$ and cycle nodes $v_1, \ldots, v_k$, we introduce two classes of subproblems.
For the 5-wheel depicted in Fig.~\ref{fig:OddWheel}, these subproblems are depicted in Fig.~\ref{fig:OddWheelTriangulation}.

For every $j \in \{2, \ldots, k\}$, we add the triangle subproblem $u v_1 v_j \in \BV$, as described in the previous section.

For every $j \in \{2, \ldots, k-1\}$, we add the subproblem $u v_j v_{j+1}, v_1 \in \BV$ for the lollipop graph that consists of the triangle $u v_j v_{j+1}$ and the additional edge $u v_1$.
The feasible set $X_{uvw,s}$ of a lollipop graph $uvw,s$ has ten elements, five feasible multicuts of the triangle times two for accounting for the additional edge.

\subsection{Dependencies}

The dependency between triangle subproblems and edge subproblems are expressed below in the form of a linear system.
It fits into thee form \eqref{eq:IPSLP-constraints} of an \IPSLP.
\begin{align*}
	\mu_{uv} & = \mu_{uvw}(1,1,0) + \mu_{uvw}(1,0,1) + \mu_{uvw}(1,1,1)\\
	\mu_{uw} & = \mu_{uvw}(1,1,0) + \mu_{uvw}(0,1,1) + \mu_{uvw}(1,1,1)\\
	\mu_{vw} & = \mu_{uvw}(1,0,1) + \mu_{uvw}(0,1,1) + \mu_{uvw}(1,1,1)
\end{align*}

The dependency between a lollipop subproblem with edge set $L = \{e_1,e_2,e_3,e_4\}$ and a triangle subproblem with edge set $T = \{e'_1,e'_2,e'_3\}$ is stated below as a linear system with sums over edges not shared between $L$ and $T$.
This linear system has the form \eqref{eq:IPSLP-constraints} of an \IPSLP.
\begin{align*}
\forall x_{L \cap T}:
\sum_{x_{L \backslash T}} \mu_{L}(x_{L \cap T}, x_{L \backslash T}) = \sum_{x_{T \backslash L}} \mu_{T}(x_{T \cap L}, x_{T \backslash C})
\end{align*}

\subsection{Remarks}

\textbf{Remark~1.} The triangulation of cycles can be understood as the constructing of a junction tree~\cite{WainwrightBook} in such a way that the minimum cost multicut problem over the cycle can be solved by dynamic programming.
The triangulation of cycles can also be understood as a tightening of an outer polyhedral relaxation of the multicut polytope:
A cycle inequality \eqref{eq:mp-cycles} defines a facet of the multicut polytope if and only if the cycle is chordless \cite{ChopraMulticut}.
By triangulating a cycle, we obtain a set of minimal chordless cycles (triangles) whose cycle inequalities together imply that of the entire cycle.

\textbf{Remark~2.} Technically, we would not have needed to include triangle subproblems for odd wheels.
Instead, we could have introduced dependencies between lollipops directly in the form of an \IPSLP.
However, by introducing triangle factors in addition and by expressing dependencies between lollipops and triangles, we couple lollipop factors from different odd wheels more tightly whenever they share the same triangles.

%
%
%

\begin{figure}
	\small
	\centering
	\begin{tikzpicture}[scale=0.8]
	\draw (-1.8,0.6) + (1*72+18:1.5) node (v1) {$v_{1}$};
	\draw (-1.8,0.6) + (2*72+18:1.5) node (v2) {$v_{2}$};
	\draw (-1.8,0.6) + (0,0) node (u)  {$u$};
	\draw (v1) -- (v2) -- (u) -- (v1);
	\draw (-1.4,0) + (1*72+18:1.5) node (v1) {$v_{1}$};
	\draw (-1.4,0) + (2*72+18:1.5) node (v2) {$v_{2}$};
	\draw (-1.4,0) + (3*72+18:1.5) node (v3) {$v_{3}$};
	\draw (-1.4,0) + (0,0) node (u)  {$u$};
	\draw (v1) -- (u) -- (v2) -- (v3) -- (u);
	
	\draw (-0.7,0) + (1*72+18:1.5) node (v1) {$v_{1}$};
	\draw (-0.7,0) + (3*72+18:1.5) node (v3) {$v_{3}$};
	\draw (-0.7,0) + (0,0) node (u)  {$u$};
	\draw (v1) to[bend left=10] (v3) -- (u) -- (v1);
	
	
	\draw (1*72+18:1.5) node (v1) {$v_{1}$};
	\draw (3*72+18:1.5) node (v3) {$v_{3}$};
	\draw (4*72+18:1.5) node (v4) {$v_{4}$};
	\draw (0,0) node (u)  {$u$};
	\draw (v1) -- (u) -- (v3) -- (v4) -- (u);
	
	\draw (0.7,0) + (1*72+18:1.5) node (v1) {$v_{1}$};
	\draw (0.7,0) + (4*72+18:1.5) node (v4) {$v_{4}$};
	\draw (0.7,0) + (0,0) node (u)  {$u$};
	\draw (v1) to[bend right=10] (v4) -- (u) -- (v1);
	\draw (1.4,0) + (1*72+18:1.5) node (v1) {$v_{1}$};
	\draw (1.4,0) + (4*72+18:1.5) node (v4) {$v_{4}$};
	\draw (1.4,0) + (5*72+18:1.5) node (v5) {$v_{5}$};
	\draw (1.4,0) + (0,0) node (u)  {$u$};
	\draw (v1) -- (u) -- (v4) -- (v5) -- (u);
	\draw (1.8,0.6) + (1*72+18:1.5) node (v1) {$v_{1}$};
	\draw (1.8,0.6) + (5*72+18:1.5) node (v5) {$v_{5}$};
	\draw (1.8,0.6) + (0,0) node (u)  {$u$};
	\draw (v1) -- (v5) -- (u) -- (v1);
	\end{tikzpicture}\\[1ex]
	\caption{
		Depicted above is a triangulation of the odd wheel from Figure~\ref{fig:OddWheel}. It consists of the triangles $uv_1v_2, uv_1v_3, uv_1v_4, uv_1v_5$ and the lollipop graphs $(uv_2v_3,v_1), (uv_3v_4,v_1), (uv_4v_5,v_1)$.
	}
	\label{fig:OddWheelTriangulation}
\end{figure}
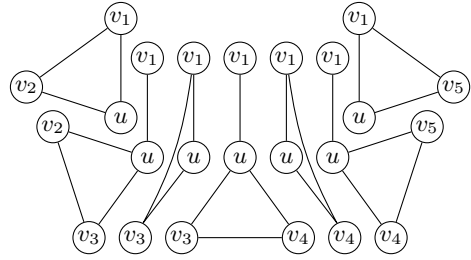

\section{Algorithm}
\label{sec:Algorithm}
We now define an algorithm for the minimum cost multicut problem \eqref{eq:mp-objective}--\eqref{eq:mp-integrality}.
This algorithm takes an instance of the problem as input and alternates for a fixed number of iterations between two main procedures.

The first procedure, defined in Sec.~\ref{sec:message-passing}, solves an instance of a dual of the \IPSLP\ relaxation defined in the previous section.
The output consists in a lower bound and a re-parameterization of the instance of the minimum cost multicut problem given as input.
The second procedure tightens the \IPSLP\ relaxation by adding subproblems for cycle inequalities \eqref{eq:mp-cycles} and odd wheel inequalities \eqref{eq:OddWheelConstraint} violated by the current solution.
Separation procedures for finding such violated inequalities, more efficiently than in cutting plane algorithms for the primal \cite{KappesMulticut,HigherOrderSegmentationByMulticuts,SungwoongHigherOrderCorrelationClustering}, are defined in Sec.~\ref{sec:separation}.

To find feasible solutions of the instance of the minimum cost multicut problem given as input, we apply a state-of-the-art local search algorithm on the computed re-parameterizations, a procedure commonly referred to as \emph{rounding} (Sec.~\ref{sec:rounding}).

\subsection{Message Passing}
\label{sec:message-passing}
Like other algorithms based on dual decomposition, the algorithm we propose does not solve the \IPSLP\ directly, in the primal domain, but optimizes a dual of \eqref{eq:IPSLP-objective}--\eqref{eq:IPSLP-constraints}.
Specifically, it operates on a space of re-parametrizations of the problem defined below:
For any two dependent subproblems $jk \in \BE$, we can change the costs $\theta_j$ and $\theta_k$ by an arbitrary vector $\Delta$ according to the \emph{update rules}
\begin{align}
\theta'_j & := \theta_j + A_{(j,k)}^\T \Delta \label{eq:update-rule-1}\\
\theta'_k & := \theta_k - A_{(k,j)}^\T \Delta \label{eq:update-rule-2}
\enspace .
\end{align}
We refer to any update of $\theta$ according to the rules \eqref{eq:update-rule-1}--\eqref{eq:update-rule-2} as \emph{message passing}.
Message passing does not change the cost of any primal feasible solution, as
\begin{align}
   & \la \theta'_j, \mu_j \ra + \la \theta'_k, \mu_k \ra  \nonumber\\
= \, & \la \theta_j + A_{(j,k)}^\T \Delta, \mu_j \ra + \la \theta_k - A_{(k,j)}^\T \Delta, \mu_k \ra\\
= \, & \la \theta_j, \mu_j \ra + \la \theta_k, \mu_k \ra + \la \Delta, A_{(j,k)} \mu_j - A_{(k,j)} \mu_k \ra \\
\overset{\eqref{eq:IPSLP-constraints}}{=} \, & \la \theta_j, \mu_j \ra + \la \theta_k, \mu_k \ra 
\enspace .
\end{align}

Message passing does, however, change the \emph{dual lower bound}  $L(\theta)$ to~\eqref{eq:IPSLP-objective} given by
\begin{align}
	L(\theta) := \sum_{j \in \BV} \min_{x \in X_i} \la \theta_j, x_j \ra
	\enspace .
\end{align}

The maximum of $L(\theta)$ over all costs obtainable by message passing is equal to the minimum of \eqref{eq:IPSLP-objective}, by linear programming duality.
We seek to alter the costs $\theta$ by means of message passing so as to maximize the lower bound $L(\theta)$.
For the general \IPSLP, a framework of algorithms to achieve this goal is defined in
\cite{ConvergentMessagePassingNIPS}.
For the minimum cost multicut problem, we define and implement Alg.~\ref{alg:multicut-message-passing} within this framework.
The specifics of this algorithm for the minimum cost multicut problem are discussed below.
General properties of message passing for \IPSLP\ s are discussed in
\cite{ConvergentMessagePassingNIPS}.

\begin{algorithm}[t]
\KwData{$\{i_1,\ldots,i_k\} = \BV$, $(\theta_{i})_{i \in \BV}$, $(A_{(j,i)},A_{(i,j)})_{ij \in \BE}$}
\For{$i = i_1,\ldots,i_k$} {
	\If{$i$ is an edge subproblem $uv$:} {
		\textbf{Receive messages:} \\
        \For{$w \in \FV: uvw \in T$} {
			$\delta := \min_{x_{uw},x_{vw}} \theta_{uvw}(1,x_{uw},x_{vw})$ \\
				\hfill $ - \min_{x_{uw},x_{vw}} \theta_{uvw}(0,x_{uw},x_{vw}) $ \\
			$\theta_{uv} \,\textnormal{+=}\, \delta$ \\
			$\forall x_{uw},x_{vw}: \, \theta_{uvw}(1,x_{uw},x_{vw}) \,\textnormal{-=}\, \delta$ \\
        }
        \textbf{Send messages:} \\
			$\delta := \abs{\{w \in \FV: uvw \in T\}}^{-1} \theta_{uv}$\\
			$\theta_{uv} := 0$\\
				\For{$w \in \FV: uvw \in T$} {
					$\forall x_{uw},x_{vw}: \, \theta_{uvw}(1,x_{uw},x_{vw}) \,\textnormal{+=}\, \delta$\\
		}
    }
	\If{$i$ is a triangle subproblem $uvw$ with edges $C$:} {
		\textbf{Receive messages:} \\
		\For{lollipops $L$ with $L\cap C \neq \varnothing$} {
			$\delta(x_{L \cap C}) := \min_{x_{L \backslash C}} \theta_{L}(x_{L \cap C},x_{L \backslash C})$\\
			$\theta_{C}(x_{L \cap C},x_{C \backslash L}) \,\textnormal{+=}\, \delta(x_{L\cap C})$\\
			$\theta_{L}(x_{L \cap C},x_{L \backslash C}) \,\textnormal{+=}\, \delta(x_{L\cap C}$\\
		}
		\textbf{Send messages:} \\
		$\alpha := \abs{ \{L \text{ a lollipop} : L \cap C \neq \varnothing\} }$\\
		\For{lollipops $L$ with $L\cap C \neq \varnothing$} {
			$\delta_{L}(x_{L \cap C}) := \min_{x_{C \backslash L}} \theta_{uvw}(x_{L \cap C},x_{C \backslash L})$\\
			$\theta_{L}(x_{L \cap C},x_{L \backslash C}) \,\textnormal{+=}\, \frac{1}{\alpha} \delta_L(x_{L\cap C})$\\
		}
		\For{lollipops $L$ with $L\cap C \neq \varnothing$} {
			$\theta_{C}(x_{L\cap C}, x_{C \backslash L}) \,\textnormal{+=}\, \frac{1}{\alpha} \delta_L(x_{L \cap C})$\\
		}
	}
}
\caption{Message passing for the multicut problem}
\label{alg:multicut-message-passing}
\end{algorithm}

\paragraph{Factor Order.}
Alg.~\ref{alg:multicut-message-passing} iterates through all edge and triangle subproblems.
The order is specified as follows:
We assume that a node order is given.
With respect to this node order, edges $uv \in E$ are ordered lexicographically.
For every triangle and its edge set $C = \{e_1,e_2,e_3\} \subseteq E$ with $e_1 < e_2 < e_3$, we define the ordering constraint $e_1 < C < e_3$.
For every lollipop graph and its edge set $L = \{e_1, e_2, e_3, e_4\}$ with $e_1 < e_2 < e_3 < e_4$, we define the ordering constraint $e_1 < L < e_4$.
The strict partial order defined by these constraints is extended to a total order by topological sorting.

\paragraph{Message Passing Description.}
When an edge subproblem $uv \in \BE$ is visited, Alg.~\ref{alg:multicut-message-passing} receives messages from all dependent triangle subproblems. 
Having received a message from triangle $uvw \in \BE$, the costs $\theta_{uvw}$ satisfy the condition
\begin{align*}
  \min_{x_{uw},x_{vw}} \!\! \theta_{uvw}(0,x_{uw},x_{vw}) 
= \min_{x_{uw},x_{vw}} \!\! \theta_{uvw}(1,x_{uw},x_{vw})
\enspace .
\end{align*}
In other words, the cost of the triangle factor $\theta_{uvw}$ has no preference for either $x_{uv} = 0$ or $x_{uv} = 1$.
Sending messages from $\theta_{uv}$ is analoguous:
Having sent messages from $uv$, we have $\theta_{uv} = 0$, i.e., there is again no preference for either $x_{uv} = 0$ or $x_{uv} = 1$.

When we visit a triangle subproblem $uvw$, we do the analogous with all dependent lollipop subproblems: 
Once messages have been received, lollipop subproblems have no preference for incident edges.
Once messages have been sent, this holds true for the triangle subproblems.

Once Alg.~\ref{alg:multicut-message-passing} has visited all subproblems and terminates, we reverse the order of subproblems and invoke Alg.~\ref{alg:multicut-message-passing} again. 
This double call of Alg.~\ref{alg:multicut-message-passing} is repeated for a fixed number of iterations that is a parameter of our algorithm.

\subsection{Separation}
\label{sec:separation}

Applying Alg.~\ref{alg:multicut-message-passing} with all cycles and all odd wheels of a graph $G$ is impractical, as the number of triangles for cycle inequalities~\eqref{eq:mp-cycles} is cubic, and the number of lollipop graphs for odd wheels~\eqref{eq:OddWheelConstraint} is quartic in $|E|$.
In order to arrive at a practical algorithm, we take a cutting plane approach in which we separate and add subproblems for violated cycle and odd wheel inequalities periodically.
Initially, $\BV$ contains only one element for every edge $e \in E$, and $\BE$ is empty.

In the primal, given some fractional $x \in [0,1]^E$, it is common to look for maximally violated inequalities~\eqref{eq:mp-cycles} and~\eqref{eq:OddWheelConstraint}.
This is possible in polynomial time via shortest path computations~\cite{ChopraMulticut,GeometryOfCutsAndMetrics}.
In our dual formulation, we have no primal solution $x$ to search for violated inequalities.
Here, a suitable criterion is to consider those additional triangle or lollipop subproblems that necessarily increase the dual lower bound $L(\theta)$ by some constant $\epsilon > 0$.
Among these subproblems, we choose those for which the increase is maximal and add them to the graph $(\BV,\BE)$.
A similar dual cutting plane approach has shown to be useful for graphical models in~\cite{FrustratedCyclesSontagEtAl}.
As we discuss below, separation is more efficient in the dual than in the primal.

\subsubsection{Cycle Inequalities}

\begin{algorithm}[t]
	\KwData{$G = (V,E), \, \, \epsilon \geq 0, \, \theta_e \in \R$}
	$l := 1$\\
	\For{$uv \in E$} {
		\If{$\theta_{uv} \geq \epsilon$} {
			$\textnormal{union}(u,v)$\\
		}
	}
	\For{$uv \in E$} {
		\If{$\theta_{uv} \leq -\epsilon$ \textbf{and} \textnormal{find}(u) = \textnormal{find}(v)} {
			$C_l := \textnormal{shortest-path}(u,v,\epsilon)$ \\
			$l := l + 1$\\
		}
	}
	\caption{Separation of cycle inequalities~\eqref{eq:mp-cycles}}
	\label{alg:FindViolatedCycle}
\end{algorithm}

We characterize those cycles whose subproblem increases the dual lower bound $L(\theta)$ by at least $\epsilon$.
\begin{proposition}
  \label{prop:CycleCharacterization}
Let $C = \{e_1,\ldots,e_k\}$ be a cycle with $\theta_{e_1} \leq -\epsilon$ and $\theta_{e_l} \leq \epsilon$ for $l > 1$.
Then, the dual lower bound $L(\theta)$ can be increased by $\epsilon$ by including a triangulation of $C$.
\end{proposition}

In order to find such cycles, we apply Alg.~\ref{alg:FindViolatedCycle}.
This algorithm first records in a disjoint set data structure whether distinct nodes $u,v \in V$ are connected via edges with weight $\geq \epsilon$.
Then, it visits all edges $e \in E$ with $\theta_e \leq - \epsilon$.
If the endpoints of $e$ are connected by a path along which all edges have weight at least $\epsilon$, it searches for a shortest such path by means of breadth first search.

In the primal, finding a maximally violated cycle inequality~\eqref{eq:mp-cycles} is more expensive, requiring, for every edge $uv \in E$, the search for a $uv$-path with minimum cost $x$~\cite{ChopraMulticut} by, e.g., Dijkstra's algorithm.

\subsubsection{Odd Wheel Inequalities}

\begin{algorithm}[t]
\KwData{Triangles $uvw$, costs $\theta_{uvw}$, $\epsilon \geq 0$} 
$l := 0$\\
\For{$u \in V$} {
	$G' = (V', E'), V' = \varnothing$, $E' = \varnothing$, $Connect = \varnothing$\\
	\For{triangles $uvw$} {
		\If{\eqref{eq:OddWheelTriangleTest} holds true}{
			$V' := V' \cup \{v,v',w,w'\}$\\
			$E' := E' \cup \{vw', v'w\}$\\
			union$(v,v')$\\
			union$(w,w')$\\
		}
	}
	\For{$v \in V' \cap V$}{
		\If{$\textnormal{find}(u) = \textnormal{find}(v)$} {
			$P' := \textnormal{shortest-path}_{G'}(u,v,\epsilon)$\\
			$C = \{uv \in E \,|\, uv' \in P' \vee u'v \in P'\}$\\
			\If{$C$ is a simple cycle in $G$}{
				$O_l := \{u,P\}$\\
				$l := l + 1$\\
			}
		}
	}
}
\caption{Separation of odd wheel inequalities~\eqref{eq:OddWheelConstraint}}
\label{alg:FindViolatedOddWheel}
\end{algorithm}

We characterize those odd wheels whose lollipop subproblem increases the lower bound $L(\theta)$ by at least $\epsilon$.

\begin{proposition}
  \label{prop:OddWheelCharacterization}
Let $O$ an odd wheel with center node $u$ and cycle nodes $v_1,\ldots,v_k$. 
Adding the lollipop subproblems for $O$ increases $L(\theta)$ by at least $\epsilon$ if the costs $\theta_{u v_i v_{i+1}}$ of each triangle $uv_iv_{i+1}$ are such that the minimal cost of any edge labeling of the triangle cutting precisely one edge incident to $u$ is smaller by $\epsilon$ than the minimal cost of any edge labeling of the triangle cutting $0$ or $2$ edges incident to $u$.
That is:
\begin{align}
& \min_{\{x: \, x_{uv_i} +x_{uv_{i+1}} = 1\}} \theta_{u v_i v_{i+1}}(x) + \epsilon \nonumber\\
\leq \quad & \min_{\{x: \, x_{uv_i} +x_{uv_{i+1}} \neq 1\}} \theta_{u v_i v_{i+1}}(x)
\enspace .
\label{eq:OddWheelTriangleTest}
\end{align}
\end{proposition}

In order to find such odd wheels, we apply Alg.~\ref{alg:FindViolatedOddWheel}.
This algorithm builds on our observation that we need to look only at triangles whose subproblem has already been added.
Hence, Alg.~\ref{alg:FindViolatedOddWheel} visits each node $u \in V$ and builds a bipartite graph $G' = (V',E')$ as follows.
(An example is depicted in Fig.~\ref{fig:OddWheelSeparationGraph} for a 5-wheel and \eqref{eq:OddWheelTriangleTest} holding true for all triangles of the wheel.)
For each triangle $u v k$ such that~\eqref{eq:OddWheelTriangleTest} holds true, four nodes $v, v', k, k' \in V'$ are added to $V'$, two copies of each original node.
These are joined by edges $uv', u'v \in E'$.
If a path from $u$ to $u'$ exists in $G'$, we have found a violated odd wheel inequality~\eqref{eq:OddWheelConstraint}.
As $G'$ is bipartite, a $uu'$-path in $G'$ corresponds to an odd cycle in $G$.
As before, the search for paths is accelerated by connectivity tests via a disjoint set data structure and is carried out by breadth first search.

In the primal, finding a maximally violated odd wheel inequality~\eqref{eq:OddWheelConstraint} entails the same construction of the bipartite graph $G'$ for each node $u \in V$~\cite{GeometryOfCutsAndMetrics}.
However, a shortest path search w.r.t.~edge costs $\frac{1}{2} - x_{vw} + \frac{1}{2}(x_{uv} + x_{uw})$ needs to be carried out by Dijkstra's algorithm instead of breadth first search.
Further complication in the primal comes from the fact that a separation algorithm needs to visit all $v \in V'$ in order to compute the shortest $vv'$-path in $G'$.

\begin{figure}
	\small
	\centering
	\begin{tikzpicture}[yscale=0.5]
	\tikzset{edge/.style={-latex}}
	
	\draw (0,0) node (v1)  {$v_1$};
	\draw (0,2) node (v1')  {$v'_1$};
	\draw (1,0) node (v2)  {$v_2$};
	\draw (1,2) node (v2')  {$v'_2$};
	\draw (2,0) node (v3)  {$v_3$};
	\draw (2,2) node (v3')  {$v'_3$};
	\draw (3,0) node (v4)  {$v_4$};
	\draw (3,2) node (v4')  {$v'_4$};
	\draw (4,0) node (v5)  {$v_5$};
	\draw (4,2) node (v5')  {$v'_5$};
	
	\draw[edge] (v1) to (v2');
	\draw[edge] (v1') to (v2);
	\draw[edge] (v2) to (v3');
	\draw[edge] (v2') to (v3);
	\draw[edge] (v3) to (v4');
	\draw[edge] (v3') to (v4);
	\draw[edge] (v4) to (v5');
	\draw[edge] (v4') to (v5);
	\draw[edge] (v5) to (v1');
	\draw[edge] (v5') to (v1);
	\end{tikzpicture}\\[1ex]
	\caption{Depicted above is the bipartite graph $G'$ constructed by Alg.~\ref{alg:FindViolatedOddWheel} for separating the 5-wheel depicted in Fig.~\ref{fig:OddWheel}.}
	\label{fig:OddWheelSeparationGraph}
\end{figure}
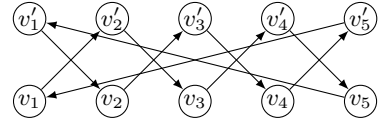

\subsection{Rounding}
\label{sec:rounding}

Our message passing Alg.~\ref{alg:multicut-message-passing} improves a dual lower bound on~\eqref{eq:mp-objective}, but does not provide a feasible solution of \eqref{eq:mp-objective}--\eqref{eq:mp-integrality}.
In order to obtain a feasible multicut, we apply a local search algorithm defined in \cite{ImageMeshDecompositionLiftedMulticut}, namely greedy additive edge contraction (GAEC), followed by Kernighan-Lin with joins (KLj).
GAEC computes a multicut by greedily contracting those edges for which the join decreases the cost maximally.
It stops as soon as no contraction of any edge strictly decreases the cost.
KLj attempts to improve a given multicut recursively by applying transformations from three classes: (1)~moving nodes between two components, (2)~moving nodes from a given component to a newly forming one or (3)~joining two components.
GAEC and KLj are local search algorithms that output a feasible multicut that need not be optimal.

We apply GAEC and KLj not only to the instance of the minimum cost multicut problem given as input but also to the re-parameterization of this instance output by Alg.~\ref{alg:multicut-message-passing}.
The rationale for doing so comes from LP duality:

\begin{proposition}
\label{prop:MulticutComplementarySlackness}
Assume $\theta$ maximizes the dual lower bound $L(\theta)$ and the relaxation is tight, i.e.
\begin{align}
L(\theta)
= \min_{\{x \in \{0,1\}^E \,|\, x^{-1}(1) \in \mathcal{M}_G\}} \la \theta, x \ra
\enspace .
\end{align}
Moreover, let $\hat x \in \{0,1\}^E$ such that $\hat x^{-1}(1)$ is an optimal multicut of $G$. Then,
\begin{align}
\theta_e \begin{cases} 
	\leq 0 & \textnormal{if}\ \hat x_e = 1 \\ 
	\geq 0, & \textnormal{if}\ \hat x_e = 0 
\end{cases}
\end{align}
\end{proposition}

Having run Alg.~\ref{alg:multicut-message-passing} for a while, we expect $\theta$ to fulfill the sign condition of Prop.~\ref{prop:MulticutComplementarySlackness} approximately.
Therefore, the sign of $\theta_e$ will be a good hint of the edge $e$ being cut.
Thus, informally, we expect local search algorithms operating on the re-parameterized instance of the problem to yield better feasible multicuts than local search algorithms operating on the given instance.

For MAP-inference in discrete graphical models, it is known from \cite{TRWSKolmogorov,SRMPKolmogorov} that primal rounding can be improved greatly when applied to cost functions re-parameterized by message passing.

\section{Experiments}
\label{sec:Experiments}

\paragraph{Solvers}
We compare against several state of the art algorithms.
\begin{itemize}
\item The algorithm \textbf{MC-ILP}~\cite{HigherOrderSegmentationByMulticuts} is an efficient implementation of a cutting plane algorithm solving~\eqref{eq:mp-objective} using cycle inequalities~\eqref{eq:mp-cycles} in a cutting plane fashion.
  CPlex~\cite{cplex} is used to solve the underlying ILP problems. The integrality conditions in~\eqref{eq:mp-integrality} are directly given to the solver. According to~\cite{HigherOrderSegmentationByMulticuts} this is beneficial due to the excellent branch and cut capabilities of CPlex~\cite{cplex}.
\item Cut, Glue \& Cut~\cite{CGC}, abbreviated as \textbf{CGC}, is a move making algorithm using planar max-cut subproblems to improve multicuts.
\item
  Fusion moves for correlation clustering~\cite{FusionMoveCC}, abbreviated as  \textbf{CC-Fusion}, fuses multicuts generated by various proposal generator with the help of auxiliary multicut problems, solved in turn by \textbf{MC-ILP}.  We use randomized hierarchical clustering and randomized watersheds as proposal generators, identified by the suffixes\textbf{-RHC} and \textbf{-RWS}. We use parameters for the proposal generators as recommended by the authors~\cite{FusionMoveCC}.
\item
  \textbf{MP-C} denotes Algorithm~\ref{alg:multicut-message-passing} when we only separate for cycle inequalities~\eqref{eq:mp-cycles} by Algorithm~\ref{alg:FindViolatedCycle}, while \textbf{MP-COW} denotes that we additionally separate for odd wheel inequalities~\eqref{eq:OddWheelConstraint} by Algorithm~\ref{alg:FindViolatedOddWheel}. We search for triangles and lollipops to add every 10th iteration.
\item \textbf{KL} is the GAEC and KLj implementation~\cite{ImageMeshDecompositionLiftedMulticut} described in Section~\ref{sec:rounding} for computing multicuts. 
  We let \textbf{KL} run every 100th iteration of \textbf{MP-C} and \textbf{MP-COW} on the current reparametrized edge costs.
\end{itemize}
\textbf{MC-ILP}, \textbf{CGC} and \textbf{CC-Fusion} are implemented as part of the OpenGM suite~\cite{OpenGMBenchmark}.
Only \textbf{MC-ILP} and our solvers \textbf{MP-C} and \textbf{MP-COW} generate dual lower bounds.
\textbf{CGC} also outputs dual lower bounds, but these are equivalent to the trivial lower bound $\sum_{e \in \FE} \min(0,\theta_e)$, where edge weights $\theta_e$ are as given by the problem.
It has been shown that \textbf{CGC}, \textbf{CC-Fusion} and \textbf{KL} outperform other primal heuristics~\cite{FusionMoveCC}, hence we do not compare to any other heuristic algorithm.
Also \textbf{MC-ILP} outperforms the LP-based solver~\cite{NowozinMulticut}, due to the latter using the slower COIN-OR CLP~\cite{CoinOrCLP} solver internally, hence we exclude it from the comparison as well.

All solvers were run on a laptop computer with a i5-5200 CPU with 2.2 GHz and 8GB RAM.

\paragraph{Datasets}
We compare on 8 datasets of diverse origin.
\begin{itemize}
\item \texttt{image-seg} consists of images of the Berkeley segmentation dataset~\cite{BerkeleySegmentationDataset}, presegmented with superpixels, for which pairwise affinity values have been computed as in~\cite{ImageSegmentationClosednessAndres}. 
\item The \texttt{knott-3d-\{150|300|450|550\}} datasets come from a neural circuit reconstruction problem of tissue~\cite{ClosedSurfaceSegmentationConnectomics} with $[150]^3$, $[300]^3$, $[450]^3$ and $[900]^3$ voxels. The data is presegmented into supervoxels.
\item \texttt{modularity clustering} aims to cluster a social network into subgroups based on affinity between individual persons.
\item \texttt{CREMI-\{small|large\}} datasets were constructed as part of the CREMI~\cite{CREMI} challenge, which aims to reconstruct neural circuits of the adult fly brain. The images are taken by electron microscopy.
  The \texttt{-small} instances are cropped versions of the \texttt{-large} ones.
To our knowledge, the \texttt{CREMI-large} dataset contain the largest multicut problems approached with LP-based methods.
\end{itemize}
The \texttt{image-seg}, \texttt{knott-3d} and \texttt{modularity clustering} datasets were taken from the OpenGM benchmark~\cite{OpenGMBenchmark},
while the \texttt{CREMI} datasets were kindly provided by their authors and are not yet published.

The dataset consists of 100, 8, 8, 8, 8, 6, 3 and 3 instances, in total 144.
Dataset details can be found in Table~\ref{Table:DatasetResults}.


\paragraph{Evaluation}
We have set a timelimit of one hour for all algorithms.
In Table~\ref{Table:DatasetResults} results averaged over all instances in specific datasets are reported.
In Figure~\ref{fig:DatasetRuntimePlot} primal solution energy and dual lower bound (where applicable) averaged over all instances in specific datasets are drawn against runtime.


As can be seen from Table~\ref{Table:DatasetResults}, except for dataset \texttt{CREMI-large}, our solver \textbf{MP-COW} gives dual bounds that are within 0.0045\%, 1.9\%, 0.0061\%, 0.0068\%, 0.0017\%, 0.0007\% and 0.0083\% of the dual lower bound obtained by \textbf{MC-ILP}, which uses the advanced branch-and-cut facilities CPlex~\cite{cplex} provides.
For \texttt{CREMI-large} only our solvers \textbf{MP-C} and \textbf{MP-COW} output dual lower bounds, as \textbf{MC-ILP} did not finish a single iteration after one hour.
As can be seen from Fig.~\ref{fig:DatasetRuntimePlot} our lower bound usually converges faster than \textbf{MC-ILP}'s.
We conjecture that \textbf{MC-C} and \textbf{MP-COW} inside a branch-and-bound solver can significantly extend the reach of exact methods for the multicut problem.


Strangely, \textbf{KL} does not perform well on \texttt{image-seg}, even though the lower bound we achieve with \textbf{MP-C} and \textbf{MP-COW} are not far from the optimal lower bounds computed by \textbf{MC-ILP}. 
On the other hand, \textbf{MP-C} and \textbf{MP-COW} give much better dual and primal results for \texttt{modularity-clustering} early on.
Generally, when compared to \textbf{MC-ILP}'s primal convergence, we give much lower values early on, and for the large-scale datasets \texttt{knott-3d-550}, \texttt{CREMI-small}, \texttt{CREMI-large}, \textbf{MCI-ILP}'s primal solutions are not useful anymore.

Unlike \textbf{MC-ILP}, our reparametrized costs can be used to improve heuristic primal algorithms.
An example of this can be seen in Fig.~\ref{fig:PrimalConvergenceExample}, where reparametrized costs improve \textbf{KL}'s solutions.

\begin{table*}
  \centering
  \scriptsize
\begin{tabular}{|lgggcccccc|}
\hline
  \multicolumn{4}{|c}{\texttt{Dataset} / \textbf{Algorithm}} & \textbf{MP-C} & \textbf{MP-COW} & \textbf{CGC} & \textbf{MC-ILP} & \textbf{CC-Fusion-RWS} & \textbf{CC-Fusion-RHC} \\ \hline
   \multirow{3}{*}{ \texttt{image-seg} } &\#I&100&UB& 4730.66  & 4732.66 & 4600.81 &       \textbf{ 4434.91 } & 4447.06 & 4436.33\\
  &\#V&$\leq3764$&LB& 4434.69 & 4434.71 & 4129.70 &       \textbf{ 4434.91} & $\ddagger$ & $\ddagger$\\ 
   &\#E&$\leq10970$&time(s)& 21.92 & 41.35 & \textbf{ 0.14 } &       11.89 & 1.19 & 1.30\\ \hline
   \multirow{3}{*}{ \texttt{\shortstack{modularity\\ clustering}} } &\#I&6&UB& \textbf{ -0.49 } & \textbf{ -0.49 } & -0.30 &       -0.44 & 0.00 & -0.44\\
&\#V&$\leq115$&LB& -0.53 & -0.53 & -0.79 &       \textbf{ -0.52 } & $\ddagger$ & $\ddagger$\\ 
   &\#E&$\leq6555$&time(s)& 0.68 & 0.80 & \textbf{ 0.15 } &       2911.10 & \textbf{ 0.00 } & 17.78\\ \hline
   \multirow{3}{*}{ \texttt{knott-3d-150} } &\#I&8&UB& -4570.61 & -4570.26 & -4220.66 &       \textbf{ -4571.69 } & -4534.76 & -4552.51\\
&\#V&$\leq972$&LB& -4572.65 & -4571.97 & -4855.18 &       \textbf{ -4571.69 } & $\ddagger$ & $\ddagger$\\ 
   &\#E&$\leq5656$&time(s)& 0.84 & 3.81 & \textbf{ 0.04 } &       2.37 & 0.26 & 0.53\\ \hline
   \multirow{3}{*}{ \texttt{knott-3d-300} } &\#I&8&UB& -27285.41 & -27285.15 & -24864.59 &       \textbf{ -27302.78 } & -27242.03 & -27247.29\\
&\#V&$\leq5896$&LB& -27307.36 & -27304.64 & -28901.58 &       \textbf{ -27302.78 } & $\ddagger$ & $\ddagger$\\ 
   &\#E&$\leq36221$&time(s)& 475.63 & 747.81 & \textbf{ 2.73 } &       227.33 & 2.96 & 8.15\\ \hline
   \multirow{3}{*}{ \texttt{knott-3d-450} } &\#I&8&UB& \textbf{ -78426.70 } & \textbf{ -78426.70 } & -70865.27 & -78391.32 & -78386.14 & -78381.06\\
&\#V&$\leq17074$&LB& -78527.23 &  -78523.89  & -83272.85 &       \textbf{ -78522.51 } & $\ddagger$ & $\ddagger$\\ 
   &\#E&$\leq107060$&time(s)& 3649.05 & 3640.66 & \textbf{ 31.56 } &       1840.47 & \textbf{ 16.52 } & 119.23\\ \hline
   \multirow{3}{*}{ \texttt{knott-3d-550} } &\#I&8&UB& -136439.78  &-136439.78  & -123841.47 &       -135766.90 & \textbf{ -136464.05 } & -136395.89\\
&\#V&$\leq31249$&LB& -136814.88 & -136803.34  & -144703.64 &       \textbf{ -136755.36 } & $\ddagger$ & $\ddagger$\\ 
   &\#E&$\leq195271$&time(s)& 3783.33 & 3745.76 & \textbf{ 102.42 } &       3683.22 & \textbf{ 72.94 } & 594.60\\ \hline
   \multirow{3}{*}{ \texttt{CREMI-small} } &\#I&3&UB& \textbf{ -213167.45 } & \textbf{ -213167.45 } & -194616.60 &       -209594.49 & -168905.17 & -213117.84\\
&\#V&$\leq35523$&LB& -213225.65  & -213225.67 & -215473.98 &       \textbf{ -213208.94 } & $\ddagger$ & $\ddagger$\\ 
   &\#E&$\leq235966$&time(s)& 3645.51 & 3661.71 & \textbf{ 319.01 } &       2775.81 & 3543.61 & 2555.48\\ \hline
   \multirow{3}{*}{ \texttt{CREMI-large} } &\#I&3&UB& \textbf{ -3886840.98 } & \textbf{ -3886840.98 } & $\dagger$ &       $\dagger$ & -3772597.37 & -3619190.20\\
&\#V&$\leq623435$&LB& \textbf{ -3892753.21 } & -3893090.30 & $\dagger$ &       $\dagger$ & $\ddagger$ & $\ddagger$\\ 
   &\#E&$\leq4172314$&time(s)& \textbf{ 3667.67 } & 3806.33 & $\dagger$ & $\dagger$ & 5978.08 & 23139.40\\ \hline
\end{tabular}
\caption{
  Primal solution energy (UB)/dual lower bound (LB)/runtime in seconds averaged over all instances of datasets.
  \#I means number of instances in dataset, \#V and \#E mean number of vertices and edges in multicut instances.
  $\dagger$ signifies method did not finish one iteration after one hour, so was excluded from comparison.
  $\ddagger$ means method does not output dual lower bound.
  \textbf{Bold} numbers signify lowest primal solution energy, highest lower bound, fastest runtime.
  }
  \label{Table:DatasetResults}
\end{table*}

\begin{figure*}%
\centering%
\begin{minipage}{0.22\textwidth}%
\vspace{0pt}%
  \includegraphics[width=1.5\textwidth]{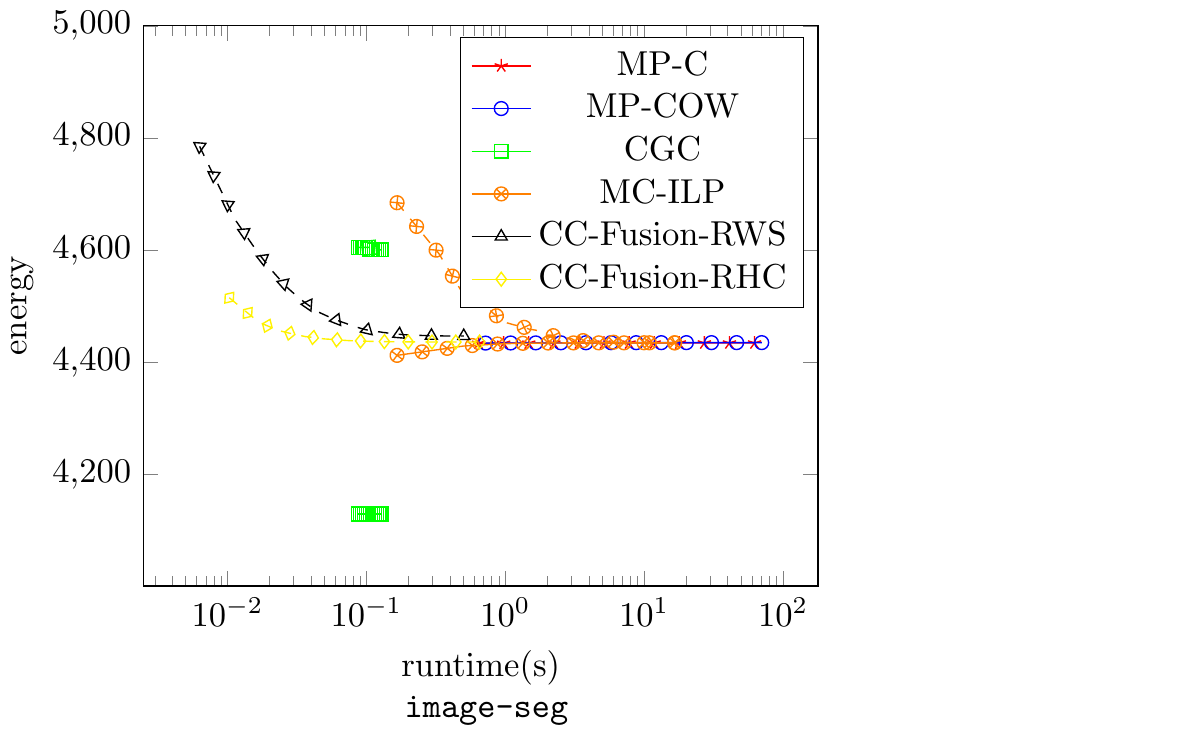}%
\end{minipage}%
\hfill
\begin{minipage}{0.22\textwidth}%
\vspace{0pt}%
  \includegraphics[width=1.5\textwidth]{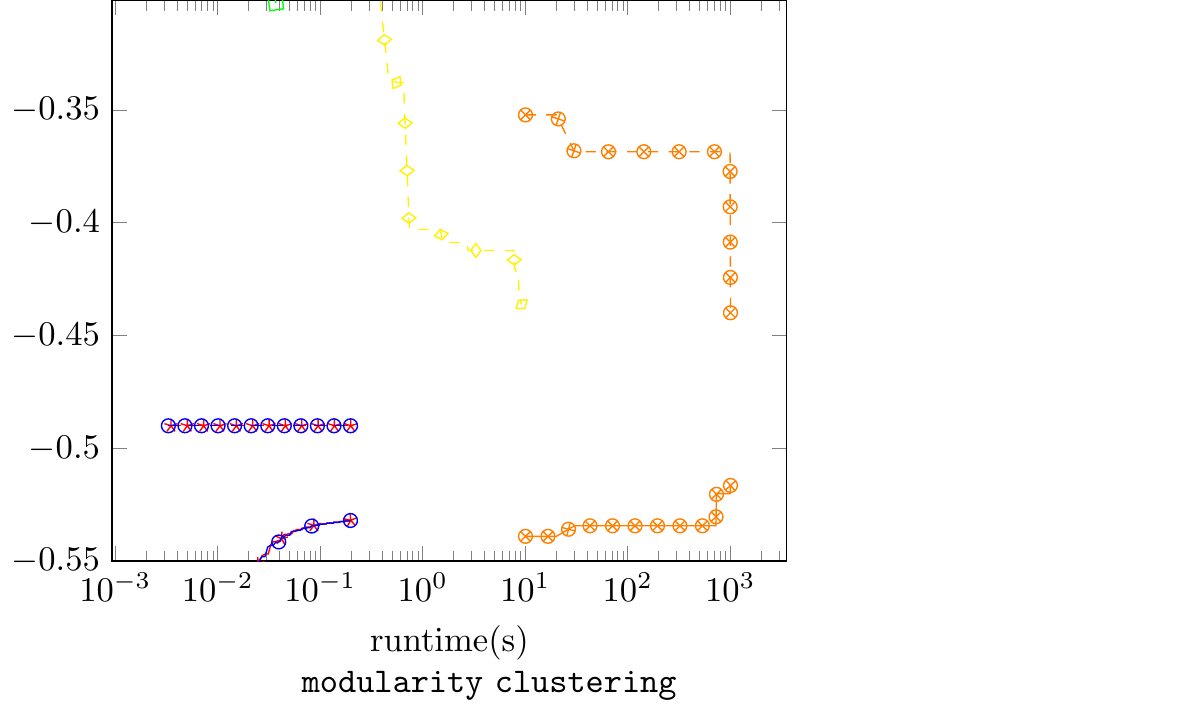}%
\end{minipage}%
\hfill
\begin{minipage}{0.22\textwidth}%
\vspace{0pt}%
  \includegraphics[width=1.5\textwidth]{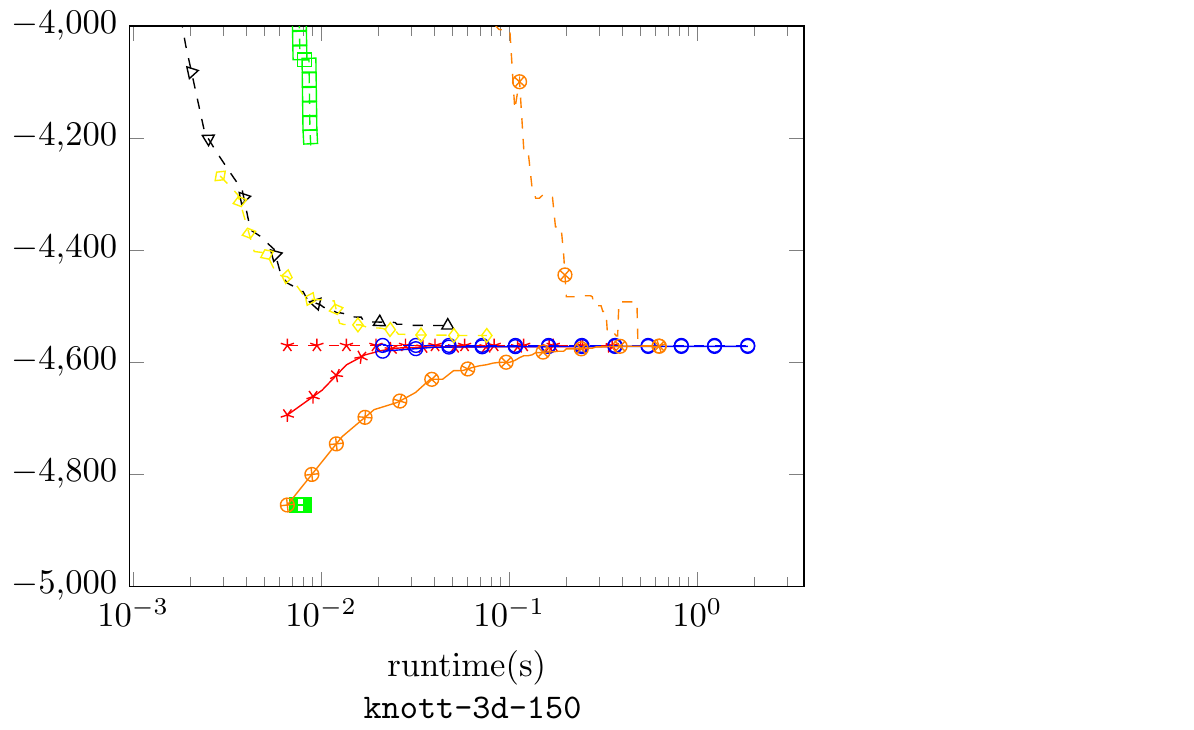}
\end{minipage}%
\hfill
\begin{minipage}{0.22\textwidth}%
\vspace{0pt}%
  \includegraphics[width=1.5\textwidth]{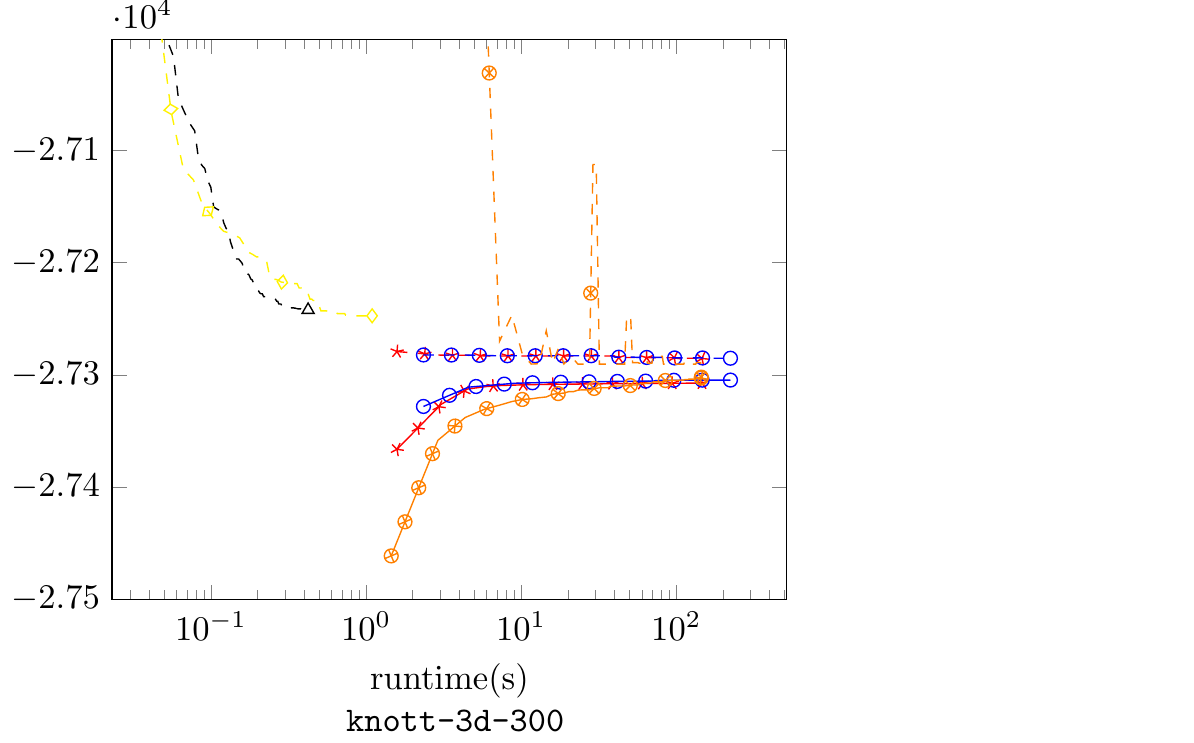}%
\end{minipage}%
\\
\begin{minipage}{0.22\textwidth}%
\vspace{0pt}%
  \includegraphics[width=1.5\textwidth]{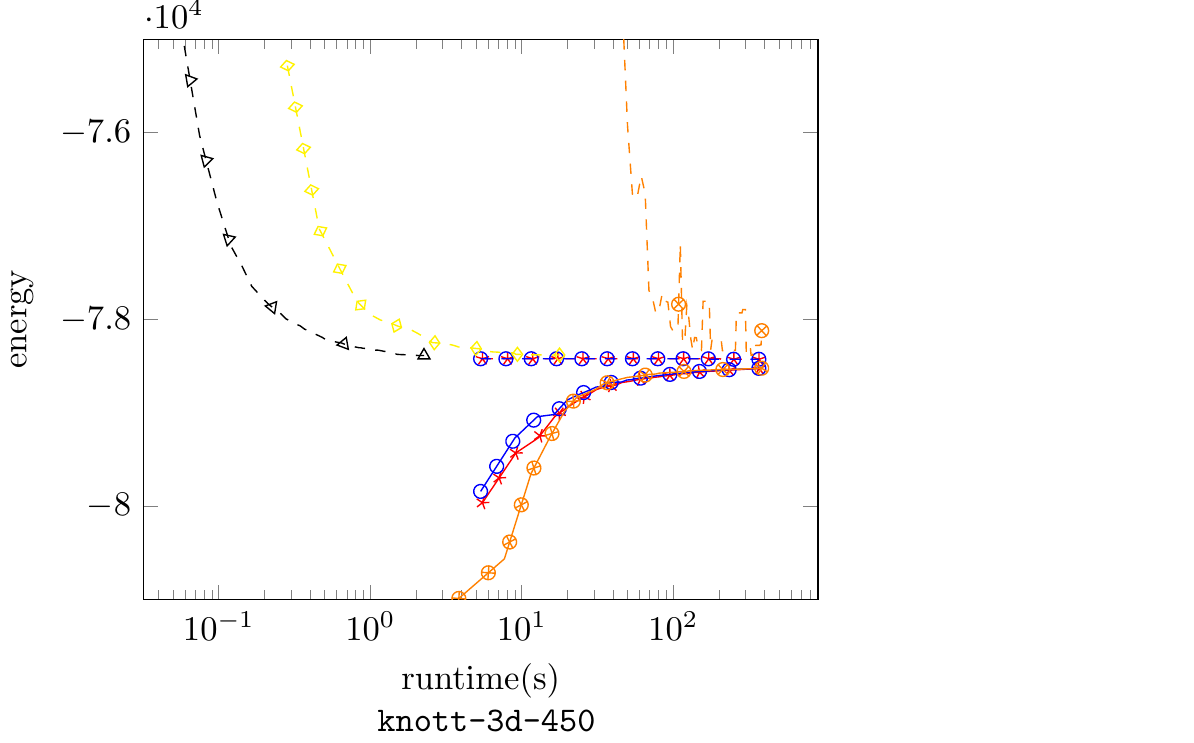}%
\end{minipage}%
\hfill
\begin{minipage}{0.22\textwidth}%
\vspace{0pt}%
  \includegraphics[width=1.5\textwidth]{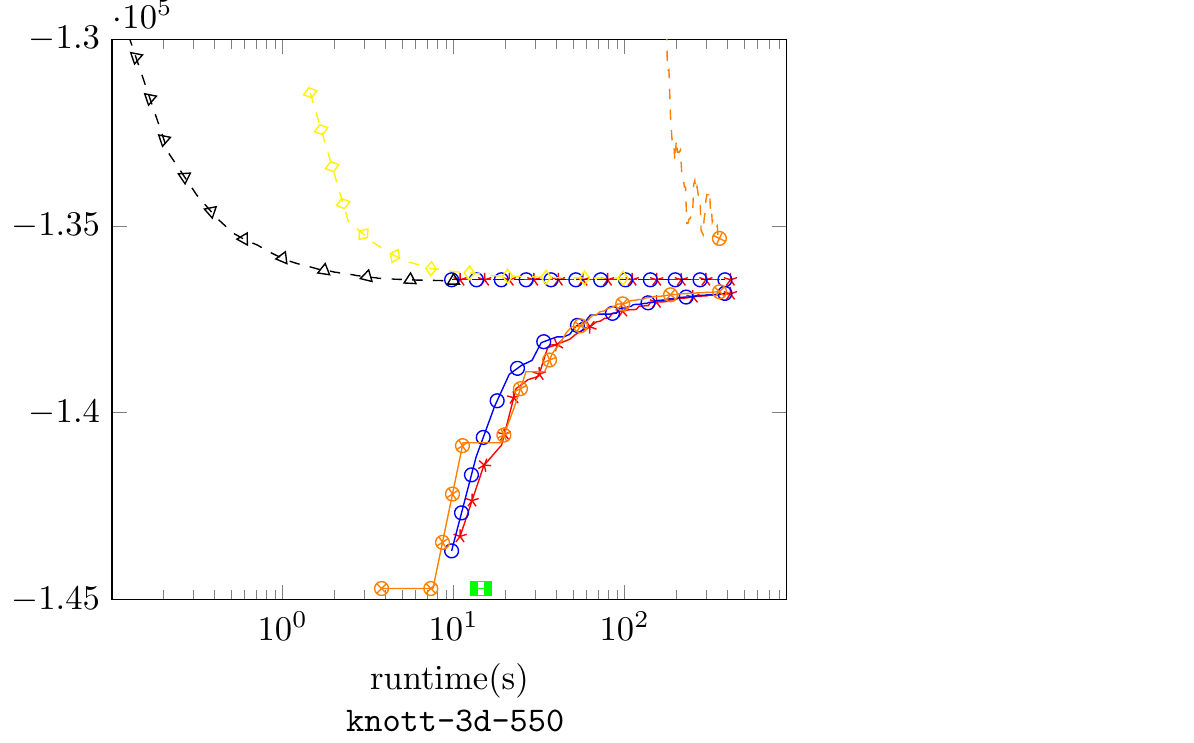}%
\end{minipage}%
\hfill
\begin{minipage}{0.22\textwidth}%
\vspace{0pt}%
  \includegraphics[width=1.5\textwidth]{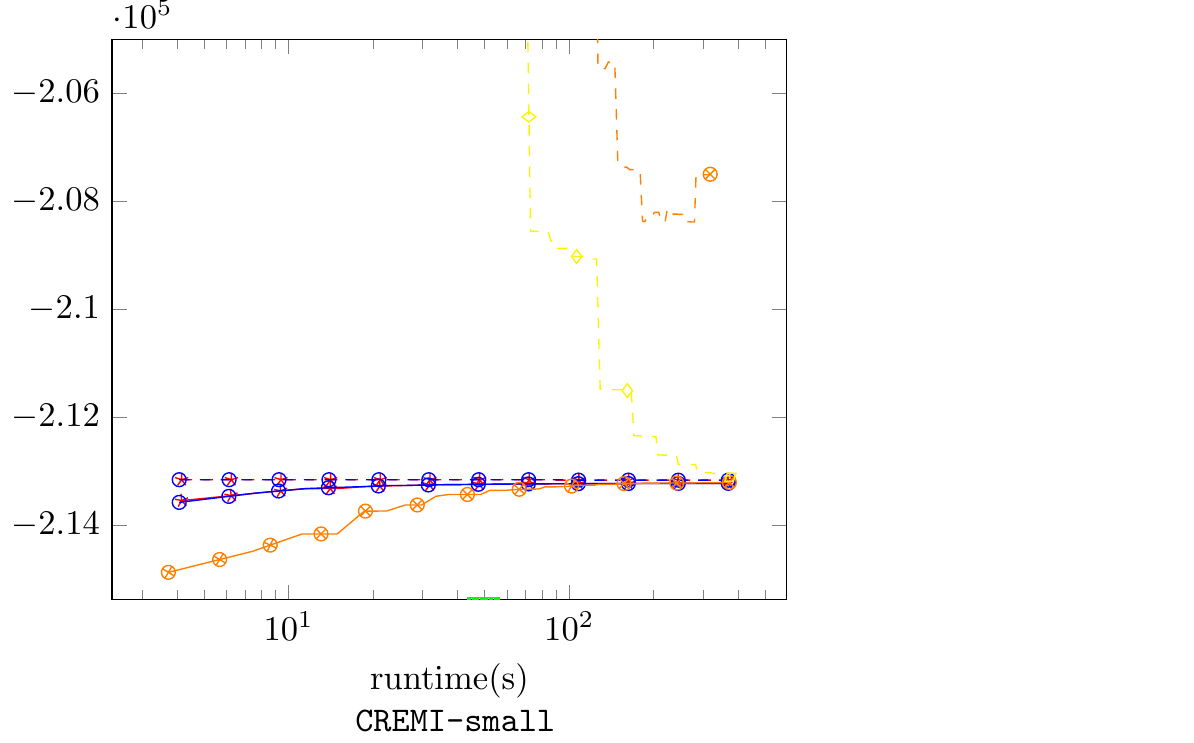}%
\end{minipage}%
\hfill
\begin{minipage}{0.22\textwidth}%
\vspace{0pt}%
  \includegraphics[width=1.5\textwidth]{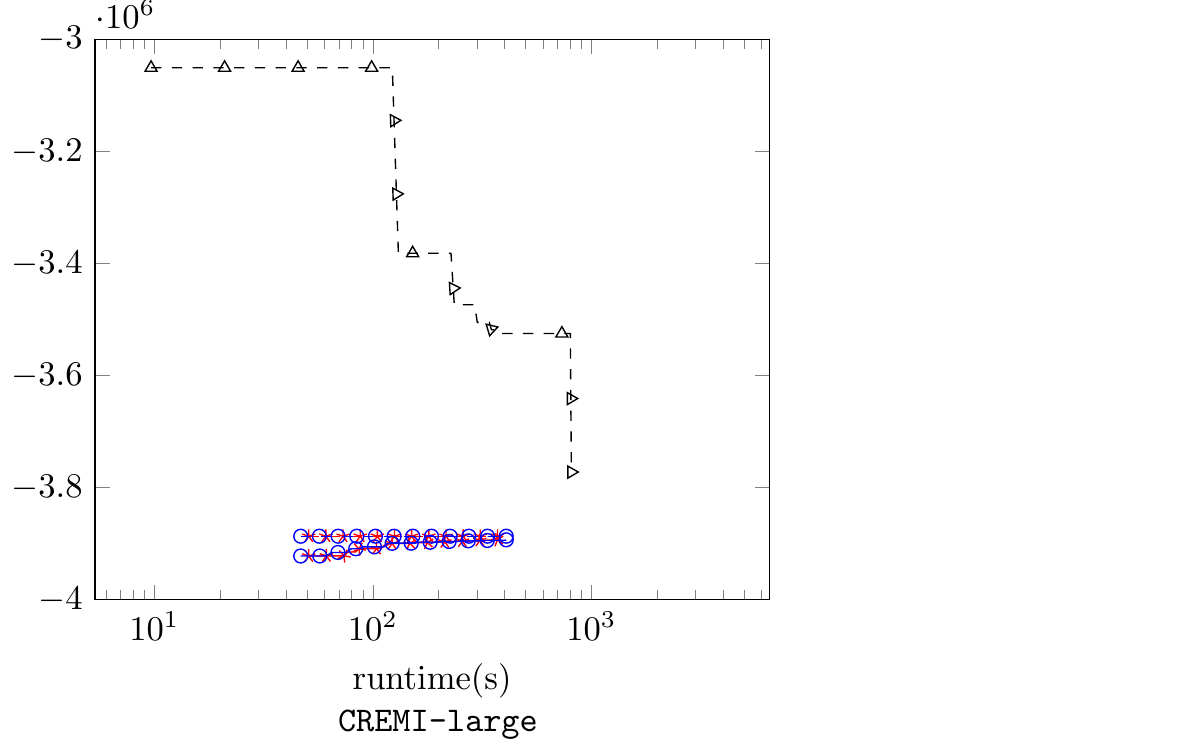}%
\end{minipage}%
  \caption{
    Averaged runtime plots for \texttt{image-seg}, \texttt{modularity clustering}, \texttt{knott-3d-150}, \texttt{knott-3d-300}, \texttt{knott-3d-450}, \texttt{knott-3d-550}, \texttt{CREMI-small} and \texttt{CREMI-large} datasets. 
 Continuous lines denote dual lower bounds and dashed ones primal energies. 
  Values are averaged over all instances of the dataset.
  The x-axis is logarithmic.
  }
\label{fig:DatasetRuntimePlot}
\end{figure*}

\begin{SCfigure}
\centering%
\includegraphics[width=0.35\textwidth]{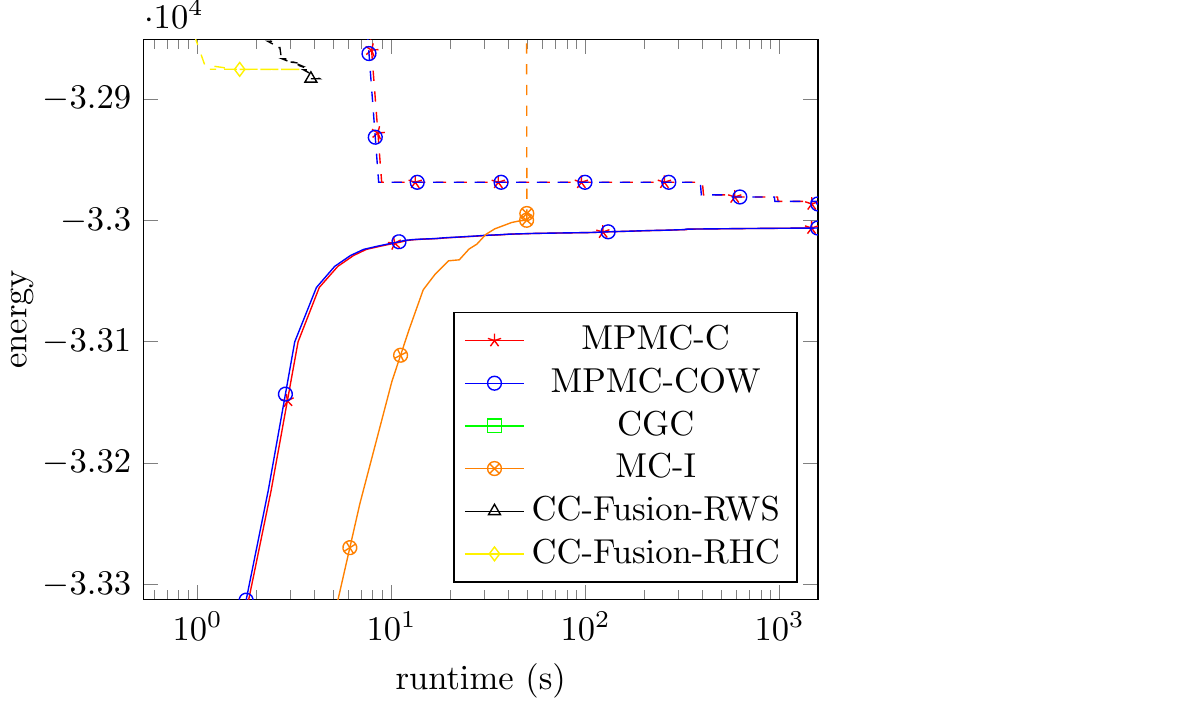}%
  \hspace*{-1.5cm}\caption{Instance \texttt{gm\_knott\_3d\_072} from dataset \texttt{knott-3d-300} where reparametrized costs improve \textbf{KL}'s solutions.}
\label{fig:PrimalConvergenceExample}
\end{SCfigure}

\paragraph{Conclusion}
We have shown that LP-based methods are feasible for solving large scale multicut problems on commodity hardware and one does not have to resort to heuristic primal algorithms.
We achieve dual bounds very close to those computed by state-of-the-art branch-and-cut solvers.
Additionally, our method usually gives much faster dual bound convergence, resulting in superior solutions when terminated early.
Also the primal heuristic GAEC + KLj can be improved when run on costs as computed by our method.

It remains an interesting task to integrate primal heuristics more tightly into our message passing approach and further improve the dual lower bound by e.g. embedding our solver into branch and cut.

\section{Acknowledgments}
The authors would like to thank Vladimir Kolmogorov for helpful discussions.
This work is partially funded by the European Research Council under the European Unions Seventh Framework Programme (FP7/2007-2013)/ERC grant agreement no 616160.

\FloatBarrier

{\small
\bibliographystyle{ieee}
\bibliography{literatur}
}

\clearpage
\section{Apendix}
\newtheorem*{proposition*}{Proposition}
\subsection{Proofs}
\paragraph{Proof of Proposition~\ref{prop:CycleCharacterization}}

\begin{proposition*}
Let $C = \{e_1,\ldots,e_k\}$ be a cycle with $\theta_{e_1} \leq -\epsilon$ and $\theta_{e_l} \geq \epsilon$ for $l > 1$.
Then, the dual lower bound $L(\theta)$ can be increased by $\epsilon$ by including a triangulation of $C$.
\end{proposition*}

\begin{proof}
Let cycle $C$ have vertices $\{v_1,\ldots,v_k\}$ and assume that $e_i \ v_i v_{i+1}$ with $k+1 = 1$ for notational purposes.
After triangulation, triangle factors on vertices $v_1v_2v_3,\ldots v_1v_{k-1}v_k$ will be present in the model.
Let the current reparametrization be $\theta$.

The triangle factors corresponding to cycle $C$ will enforce the cycle inequality~\eqref{eq:mp-cycles}
\begin{equation}\label{eq:ViolatedCycle}
x_{e_1} \leq \sum_{i=1,\ldots,k} x_{e_i}\,.
\end{equation}
It holds that
\begin{multline}
- \theta_{e_1} = \min_{x_{e_1},\ldots,x_{e_k} \in [0,1]} \sum_{i=1}^k \theta_{e_i} x_{e_i} \\
\leq -\epsilon +\min_{x_{e_1},\ldots,x_{e_k} \in [0,1]} \sum_{i=1}^k \theta_{e_i} x_{e_i} \text{ s.t.~\eqref{eq:ViolatedCycle}} \\
\leq \epsilon + \max_{\left\{ 
 \substack{\theta_{e_1},\ldots,\theta_{e_k} \\ \theta_{v_1v_2v_3},\ldots,\theta_{v_1v_{k-1}v_k} \\ \text{a reparametrization} } 
\right\} } L_C(\theta) 
\end{multline}
where $L_C(\theta) = \sum_{i=1}^k \min(0,\theta_{e_i}) + \sum_{i=2}^{k-1} \min\{\theta_{v_1v_iv_{i+1}}\}$ the dual lower bound on cycle $C$.
The first inequality above is due to either $x_{e_1} = 0$ in the optimal solution or one $x_{e_2},\ldots,x_{e_k}$ being one due to~\eqref{eq:ViolatedCycle}.
The second inequality is due to the fact that 
(i)~$\max_{\theta \text{ a reparametrization}} L(\theta) =\ min_{\mu \in \Lambda} \la \theta,\mu\ra$ by linear programming duality and
(ii)~the triangle factors enforce more inequalities than only~\eqref{eq:ViolatedCycle}.
\end{proof}

\paragraph{Proof of Proposition~\ref{prop:OddWheelCharacterization}}
\begin{proposition*}
Let $O$ an odd wheel with center node $u$ and cycle nodes $v_1,\ldots,v_k$. 
Adding the lollipop subproblems for $O$ increases $L(\theta)$ by at least $\epsilon$ if the costs $\theta_{u v_i v_{i+1}}$ of each triangle $uv_iv_{i+1}$ are such that the minimal cost of any edge labeling of the triangle cutting precisely one edge incident to $u$ is smaller by $\epsilon$ than the minimal cost of any edge labeling of the triangle cutting $0$ or $2$ edges incident to $u$.
That is:
\begin{align}
& \min_{\{x: \, x_{uv_i} +x_{uv_{i+1}} = 1\}} \theta_{u v_i v_{i+1}}(x) + \epsilon \nonumber\\
\leq \quad & \min_{\{x: \, x_{uv_i} +x_{uv_{i+1}} \neq 1\}} \theta_{u v_i v_{i+1}}(x)
\enspace .
\end{align}
\end{proposition*}

\begin{proof}
Condition~\eqref{eq:OddWheelTriangleTest} means that in all triangles in the odd wheel $O$, the minimal assignment with regard to the current reparametrization, has exactl one edge incident to $u$.
All other assignment have cost greater by at least $\epsilon$.
As $k$ is odd, there is no possiblity to combine those local assignments to a global assignment on $O$.

On the other hand, our construction of lollipop factors ensures exactness on odd wheels. As at least one triangle must then be assigned costs that are not locally optimal and which is larger by $\epsilon$ than its minimal reparametrized cost, the result follows.
\end{proof}

\paragraph{Proof of Proposition~\ref{prop:MulticutComplementarySlackness}}
\begin{proposition*}
Assume $\theta$ maximizes the dual lower bound $L(\theta)$ and the relaxation is tight, i.e.
\begin{align}
L(\theta)
= \min_{\{x \in \{0,1\}^E \,|\, x^{-1}(1) \in \mathcal{M}_G\}} \la \theta, x \ra
\enspace .
\end{align}
Moreover, let $\hat x \in \{0,1\}^E$ such that $\hat x^{-1}(1)$ is an optimal multicut of $G$. Then,
\begin{align}
\theta_e \begin{cases} 
	\leq 0 & \textnormal{if}\ \hat x_e = 1 \\ 
	\geq 0, & \textnormal{if}\ \hat x_e = 0 
\end{cases}
\end{align}
\end{proposition*}

\begin{proof}
Follows from the complementary slackness conditions in linear programming duality.
\end{proof}

\subsection{Detailed experimental evaluation}
In Table~\ref{table:instance-evaluation} a detailed per instance evaluation of all algorithms considered in the experimental section can be found. 

{\onecolumn\footnotesize
\begin{longtable}{|llcccccc|}
  \caption{
    Per instance evaluation of the considered eight datasets.
    UB means primal solution energy, LB dual lower bound and runtime(s) the runtime in seconds. 
    \textbf{Bold} numbers indicate lowest primal energy, highest lower bound and smallest runtime.
  $\dagger$ means method not applicable.
}\\
  \label{table:instance-evaluation}
  \texttt{Instance} & & \textbf{MP-C} & \textbf{MC-COW} & \textbf{CGC} & \textbf{MC-ILP} & \textbf{CC-Fusion-RWS} & \textbf{CC-Fusion-RHC} \\ \hline \endhead
\multicolumn{ 8 }{| c |}{ \textbf{ image-seg } } \\ \hline 
   \multirow{3}{*}{ 101087.bmp } & UB & 2906.16 & 2906.16 & 2853.56 & \textbf{ 2789.90 } & 2800.22 & \textbf{ 2789.90 } \\*
      & LB & 2788.69 & 2788.95 & 2622.38 & \textbf{ 2789.90 } & $\dagger$ & $\dagger$ \\*
      & runtime(s) & 0.31 & 0.59 & \textbf{ 0.02 } & 5.11 & 0.63 & 0.81 \\ \hline
   \multirow{3}{*}{ 102061.bmp } & UB & 3017.57 & 3017.57 & 3090.33 & \textbf{ 2943.77 } & 2963.42 & 2944.46 \\*
      & LB & 2932.80 & 2933.39 & 2750.99 & \textbf{ 2943.77 } & $\dagger$ & $\dagger$ \\*
      & runtime(s) & 3.26 & 5.83 & \textbf{ 0.05 } & 8.75 & 0.61 & 0.96 \\ \hline
   \multirow{3}{*}{ 103070.bmp } & UB & 4437.16 & 4444.27 & 4457.13 & \textbf{ 4199.38 } & 4205.03 & 4200.64 \\*
      & LB & 4196.88 & 4196.94 & 3842.84 & \textbf{ 4199.38 } & $\dagger$ & $\dagger$ \\*
      & runtime(s) & 11.84 & 15.93 & \textbf{ 0.15 } & 6.97 & 1.32 & 0.98 \\ \hline
   \multirow{3}{*}{ 105025.bmp } & UB & 6332.73 & 6333.88 & 6290.71 & \textbf{ 6055.33 } & 6070.84 & 6061.05 \\*
      & LB & 6045.21 & 6046.08 & 5506.01 & \textbf{ 6055.33 } & $\dagger$ & $\dagger$ \\*
      & runtime(s) & 33.75 & 55.69 & \textbf{ 0.35 } & 32.15 & 1.91 & 1.59 \\ \hline
   \multirow{3}{*}{ 106024.bmp } & UB & 1832.16 & 1832.16 & 1654.27 & \textbf{ 1599.25 } & 1618.18 & 1599.83 \\*
      & LB & 1597.07 & 1597.21 & 1466.60 & \textbf{ 1599.25 } & $\dagger$ & $\dagger$ \\*
      & runtime(s) & 0.64 & 4.09 & \textbf{ 0.02 } & 4.76 & 0.18 & 0.25 \\ \hline
   \multirow{3}{*}{ 108005.bmp } & UB & 6839.76 & 6841.08 & 6855.33 & \textbf{ 6578.03 } & 6584.62 & 6578.18 \\*
      & LB & 6567.48 & 6569.63 & 6151.29 & \textbf{ 6578.03 } & $\dagger$ & $\dagger$ \\*
      & runtime(s) & 4.58 & 22.38 & \textbf{ 0.26 } & 10.86 & 2.49 & 1.54 \\ \hline
   \multirow{3}{*}{ 108070.bmp } & UB & 9082.04 & 9083.42 & 8612.32 & \textbf{ 8422.24 } & 8445.73 & 8424.36 \\*
      & LB & 8414.64 & 8414.99 & 7818.70 & \textbf{ 8422.24 } & $\dagger$ & $\dagger$ \\*
      & runtime(s) & 28.23 & 52.32 & \textbf{ 0.41 } & 26.67 & 3.35 & 1.75 \\ \hline
   \multirow{3}{*}{ 108082.bmp } & UB & 5125.72 & 5127.99 & 5090.88 & \textbf{ 4800.15 } & 4815.78 & 4806.04 \\*
      & LB & 4786.46 & 4789.19 & 4380.66 & \textbf{ 4800.15 } & $\dagger$ & $\dagger$ \\*
      & runtime(s) & 6.85 & 25.07 & \textbf{ 0.16 } & 13.51 & 1.20 & 1.89 \\ \hline
   \multirow{3}{*}{ 109053.bmp } & UB & 4575.76 & 4579.97 & 4616.61 & \textbf{ 4421.13 } & 4424.05 & \textbf{ 4421.13 } \\*
      & LB & 4412.45 & 4413.08 & 4021.22 & \textbf{ 4421.13 } & $\dagger$ & $\dagger$ \\*
      & runtime(s) & 20.74 & 71.75 & \textbf{ 0.17 } & 9.64 & 1.50 & 0.84 \\ \hline
   \multirow{3}{*}{ 119082.bmp } & UB & \textbf{ 4512.04 } & \textbf{ 4512.04 } & 4642.96 & 4530.71 & 4535.85 & 4532.29 \\*
      & LB & 4526.91 & 4526.91 & 4346.24 & \textbf{ 4530.71 } & $\dagger$ & $\dagger$ \\*
      & runtime(s) & 0.10 & 0.10 & \textbf{ 0.02 } & 0.48 & 1.53 & 1.57 \\ \hline
   \multirow{3}{*}{ 12084.bmp } & UB & 7502.34 & 7502.34 & 7443.28 & \textbf{ 7284.45 } & 7301.17 & 7287.68 \\*
      & LB & 7276.26 & 7277.38 & 6941.02 & \textbf{ 7284.45 } & $\dagger$ & $\dagger$ \\*
      & runtime(s) & 3.48 & 8.10 & \textbf{ 0.15 } & 2.47 & 2.20 & 3.99 \\ \hline
   \multirow{3}{*}{ 123074.bmp } & UB & 4200.21 & 4204.24 & 4031.03 & \textbf{ 3842.74 } & 3856.82 & 3847.83 \\*
      & LB & 3829.51 & 3829.04 & 3439.47 & \textbf{ 3842.74 } & $\dagger$ & $\dagger$ \\*
      & runtime(s) & 35.27 & 24.55 & \textbf{ 0.14 } & 23.01 & 0.47 & 0.86 \\ \hline
   \multirow{3}{*}{ 126007.bmp } & UB & 2756.86 & 2760.64 & 2747.72 & \textbf{ 2684.83 } & 2706.76 & 2685.26 \\*
      & LB & 2677.02 & 2677.08 & 2512.08 & \textbf{ 2684.83 } & $\dagger$ & $\dagger$ \\*
      & runtime(s) & 0.23 & 0.53 & \textbf{ 0.01 } & 0.79 & 0.38 & 0.82 \\ \hline
   \multirow{3}{*}{ 130026.bmp } & UB & 6066.91 & 6066.91 & 5580.58 & \textbf{ 5350.83 } & 5369.95 & 5354.31 \\*
      & LB & 5331.00 & 5336.93 & 4828.82 & \textbf{ 5350.83 } & $\dagger$ & $\dagger$ \\*
      & runtime(s) & 36.91 & 167.54 & \textbf{ 0.26 } & 19.66 & 0.99 & 1.40 \\ \hline
   \multirow{3}{*}{ 134035.bmp } & UB & 6840.69 & 6840.69 & 6679.89 & \textbf{ 6578.98 } & 6595.87 & 6579.62 \\*
      & LB & 6562.20 & 6565.03 & 6166.95 & \textbf{ 6578.98 } & $\dagger$ & $\dagger$ \\*
      & runtime(s) & 8.44 & 43.34 & \textbf{ 0.19 } & 28.82 & 1.81 & 1.33 \\ \hline
   \multirow{3}{*}{ 14037.bmp } & UB & 1566.34 & 1582.13 & 1431.56 & \textbf{ 1383.14 } & 1393.66 & \textbf{ 1383.14 } \\*
      & LB & 1375.55 & 1375.55 & 1274.27 & \textbf{ 1383.14 } & $\dagger$ & $\dagger$ \\*
      & runtime(s) & 0.09 & 0.17 & \textbf{ 0.01 } & 0.25 & 0.13 & 0.22 \\ \hline
   \multirow{3}{*}{ 143090.bmp } & UB & 1728.22 & 1728.22 & 1807.41 & \textbf{ 1714.38 } & 1725.88 & 1715.76 \\*
      & LB & 1712.44 & 1712.26 & 1595.54 & \textbf{ 1714.38 } & $\dagger$ & $\dagger$ \\*
      & runtime(s) & 0.43 & 2.49 & \textbf{ 0.01 } & 0.56 & 0.44 & 0.34 \\ \hline
   \multirow{3}{*}{ 145086.bmp } & UB & 3806.23 & 3808.75 & 3407.83 & \textbf{ 3322.21 } & 3329.14 & 3322.59 \\*
      & LB & 3319.36 & 3319.34 & 3197.53 & \textbf{ 3322.21 } & $\dagger$ & $\dagger$ \\*
      & runtime(s) & 0.32 & 0.45 & \textbf{ 0.01 } & 0.83 & 0.41 & 1.59 \\ \hline
   \multirow{3}{*}{ 147091.bmp } & UB & 4092.25 & 4092.25 & 4129.72 & \textbf{ 3973.71 } & 3982.30 & 3975.15 \\*
      & LB & 3968.35 & 3970.58 & 3734.67 & \textbf{ 3973.71 } & $\dagger$ & $\dagger$ \\*
      & runtime(s) & 1.50 & 10.96 & \textbf{ 0.10 } & 13.24 & 0.90 & 0.83 \\ \hline
   \multirow{3}{*}{ 148026.bmp } & UB & 8411.55 & 8411.55 & 8436.70 & \textbf{ 8205.98 } & 8226.20 & 8207.72 \\*
      & LB & 8198.62 & 8199.84 & 7780.68 & \textbf{ 8205.98 } & $\dagger$ & $\dagger$ \\*
      & runtime(s) & 6.48 & 29.65 & \textbf{ 0.15 } & 4.49 & 3.10 & 2.73 \\ \hline
   \multirow{3}{*}{ 148089.bmp } & UB & 6680.96 & 6682.82 & 6666.83 & \textbf{ 6439.58 } & 6455.48 & 6440.33 \\*
      & LB & 6432.06 & 6431.52 & 6030.94 & \textbf{ 6439.58 } & $\dagger$ & $\dagger$ \\*
      & runtime(s) & 10.44 & 15.57 & \textbf{ 0.17 } & 18.64 & 2.00 & 1.53 \\ \hline
   \multirow{3}{*}{ 156065.bmp } & UB & 5798.42 & 5801.59 & 5429.45 & \textbf{ 5234.15 } & 5248.18 & 5234.76 \\*
      & LB & 5224.22 & 5225.10 & 4857.14 & \textbf{ 5234.15 } & $\dagger$ & $\dagger$ \\*
      & runtime(s) & 10.19 & 35.24 & \textbf{ 0.15 } & 18.51 & 1.49 & 1.12 \\ \hline
   \multirow{3}{*}{ 157055.bmp } & UB & 4797.86 & 4798.32 & 4768.87 & \textbf{ 4685.17 } & 4696.42 & \textbf{ 4685.17 } \\*
      & LB & 4679.04 & 4679.39 & 4472.39 & \textbf{ 4685.17 } & $\dagger$ & $\dagger$ \\*
      & runtime(s) & 0.49 & 3.10 & \textbf{ 0.02 } & 2.60 & 1.37 & 1.48 \\ \hline
   \multirow{3}{*}{ 159008.bmp } & UB & 4688.87 & 4694.40 & 4814.85 & \textbf{ 4540.87 } & 4569.12 & 4541.22 \\*
      & LB & 4534.88 & 4535.76 & 4217.95 & \textbf{ 4540.87 } & $\dagger$ & $\dagger$ \\*
      & runtime(s) & 3.53 & 21.00 & \textbf{ 0.10 } & 16.64 & 0.83 & 1.29 \\ \hline
   \multirow{3}{*}{ 160068.bmp } & UB & 3263.91 & 3265.02 & 3264.31 & \textbf{ 3089.32 } & 3103.23 & \textbf{ 3089.32 } \\*
      & LB & 3088.23 & 3088.17 & 2866.45 & \textbf{ 3089.32 } & $\dagger$ & $\dagger$ \\*
      & runtime(s) & 1.18 & 2.28 & \textbf{ 0.05 } & 1.80 & 0.61 & 1.07 \\ \hline
   \multirow{3}{*}{ 16077.bmp } & UB & 4440.51 & 4443.05 & 4408.62 & \textbf{ 4227.88 } & 4236.46 & 4228.78 \\*
      & LB & 4224.10 & 4224.60 & 3921.75 & \textbf{ 4227.88 } & $\dagger$ & $\dagger$ \\*
      & runtime(s) & 4.37 & 13.81 & \textbf{ 0.07 } & 3.22 & 1.18 & 1.63 \\ \hline
   \multirow{3}{*}{ 163085.bmp } & UB & 4493.17 & 4493.17 & 4577.62 & \textbf{ 4381.13 } & 4406.98 & 4384.59 \\*
      & LB & 4370.36 & 4371.00 & 3983.52 & \textbf{ 4381.13 } & $\dagger$ & $\dagger$ \\*
      & runtime(s) & 21.91 & 35.30 & \textbf{ 0.15 } & 9.82 & 0.91 & 1.30 \\ \hline
   \multirow{3}{*}{ 167062.bmp } & UB & 1623.20 & 1623.20 & 1281.48 & \textbf{ 1273.72 } & 1275.87 & 1275.67 \\*
      & LB & 1273.01 & 1273.29 & 1233.39 & \textbf{ 1273.72 } & $\dagger$ & $\dagger$ \\*
      & runtime(s) & 0.10 & 0.66 & \textbf{ 0.00 } & 1.00 & 0.11 & 0.17 \\ \hline
   \multirow{3}{*}{ 167083.bmp } & UB & 8979.42 & 8979.42 & 8545.37 & \textbf{ 8331.63 } & 8344.35 & 8331.90 \\*
      & LB & 8325.80 & 8328.03 & 7921.06 & \textbf{ 8331.63 } & $\dagger$ & $\dagger$ \\*
      & runtime(s) & 9.83 & 53.28 & \textbf{ 0.23 } & 11.74 & 2.27 & 1.80 \\ \hline
   \multirow{3}{*}{ 170057.bmp } & UB & 3602.31 & 3602.31 & 3355.95 & \textbf{ 3266.17 } & 3273.20 & 3266.73 \\*
      & LB & 3260.19 & 3260.98 & 2989.38 & \textbf{ 3266.17 } & $\dagger$ & $\dagger$ \\*
      & runtime(s) & 21.07 & 26.30 & \textbf{ 0.07 } & 18.20 & 0.53 & 0.67 \\ \hline
   \multirow{3}{*}{ 175032.bmp } & UB & 11863.40 & 11863.40 & 11926.00 & \textbf{ 11542.63 } & 11566.74 & 11547.67 \\*
      & LB & 11525.63 & 11526.57 & 10543.16 & \textbf{ 11542.63 } & $\dagger$ & $\dagger$ \\*
      & runtime(s) & 623.66 & 632.98 & \textbf{ 1.27 } & 165.83 & 3.25 & 3.98 \\ \hline
   \multirow{3}{*}{ 175043.bmp } & UB & 8022.65 & 8033.50 & 8224.34 & \textbf{ 7816.92 } & 7844.49 & 7822.01 \\*
      & LB & 7809.60 & 7811.17 & 7136.44 & \textbf{ 7816.92 } & $\dagger$ & $\dagger$ \\*
      & runtime(s) & 4.74 & 29.83 & \textbf{ 0.39 } & 13.17 & 4.30 & 3.46 \\ \hline
   \multirow{3}{*}{ 182053.bmp } & UB & 3700.96 & 3700.96 & 3714.59 & \textbf{ 3579.24 } & 3602.74 & 3582.99 \\*
      & LB & 3575.85 & 3576.12 & 3321.74 & \textbf{ 3579.24 } & $\dagger$ & $\dagger$ \\*
      & runtime(s) & 4.05 & 12.83 & \textbf{ 0.07 } & 14.11 & 0.65 & 0.92 \\ \hline
   \multirow{3}{*}{ 189080.bmp } & UB & 1152.76 & 1155.19 & 1095.38 & \textbf{ 1077.47 } & 1086.93 & 1078.41 \\*
      & LB & 1072.82 & 1073.00 & 972.41 & \textbf{ 1077.47 } & $\dagger$ & $\dagger$ \\*
      & runtime(s) & 0.07 & 0.11 & \textbf{ 0.00 } & 0.25 & 0.11 & 0.49 \\ \hline
   \multirow{3}{*}{ 19021.bmp } & UB & 4717.39 & 4718.32 & 4693.10 & \textbf{ 4515.08 } & 4521.63 & 4515.27 \\*
      & LB & 4498.74 & 4499.79 & 4178.50 & \textbf{ 4515.08 } & $\dagger$ & $\dagger$ \\*
      & runtime(s) & 1.67 & 4.36 & \textbf{ 0.08 } & 12.93 & 1.16 & 1.41 \\ \hline
   \multirow{3}{*}{ 196073.bmp } & UB & 592.63 & 594.41 & 572.73 & \textbf{ 545.47 } & 548.39 & \textbf{ 545.47 } \\*
      & LB & 544.89 & 544.89 & 508.04 & \textbf{ 545.47 } & $\dagger$ & $\dagger$ \\*
      & runtime(s) & 0.12 & 0.55 & \textbf{ 0.00 } & 0.76 & 0.07 & 0.24 \\ \hline
   \multirow{3}{*}{ 197017.bmp } & UB & 3207.33 & 3207.33 & 2857.86 & \textbf{ 2798.77 } & 2801.34 & \textbf{ 2798.77 } \\*
      & LB & 2792.60 & 2792.59 & 2663.98 & \textbf{ 2798.77 } & $\dagger$ & $\dagger$ \\*
      & runtime(s) & 0.82 & 1.01 & \textbf{ 0.01 } & 0.77 & 0.38 & 0.77 \\ \hline
   \multirow{3}{*}{ 208001.bmp } & UB & 6533.30 & 6533.30 & 6605.51 & \textbf{ 6272.68 } & 6277.73 & 6275.37 \\*
      & LB & 6254.97 & 6256.22 & 5773.37 & \textbf{ 6272.68 } & $\dagger$ & $\dagger$ \\*
      & runtime(s) & 22.20 & 63.32 & \textbf{ 0.25 } & 78.10 & 2.12 & 1.69 \\ \hline
   \multirow{3}{*}{ 210088.bmp } & UB & 1954.78 & 1956.20 & 2034.65 & \textbf{ 1895.44 } & 1901.54 & 1896.58 \\*
      & LB & 1891.95 & 1891.95 & 1726.75 & \textbf{ 1895.44 } & $\dagger$ & $\dagger$ \\*
      & runtime(s) & 0.26 & 0.38 & \textbf{ 0.02 } & 1.26 & 0.49 & 0.51 \\ \hline
   \multirow{3}{*}{ 21077.bmp } & UB & 3031.01 & 3032.66 & 3001.74 & \textbf{ 2946.71 } & 2966.82 & \textbf{ 2946.71 } \\*
      & LB & 2944.67 & 2944.60 & 2793.10 & \textbf{ 2946.71 } & $\dagger$ & $\dagger$ \\*
      & runtime(s) & 1.42 & 3.86 & \textbf{ 0.01 } & 4.60 & 0.44 & 1.13 \\ \hline
   \multirow{3}{*}{ 216081.bmp } & UB & 4462.75 & 4463.13 & 4324.17 & \textbf{ 4158.73 } & 4176.92 & \textbf{ 4158.73 } \\*
      & LB & 4155.98 & 4156.01 & 3952.38 & \textbf{ 4158.73 } & $\dagger$ & $\dagger$ \\*
      & runtime(s) & 0.34 & 0.60 & \textbf{ 0.03 } & 1.37 & 1.13 & 1.08 \\ \hline
   \multirow{3}{*}{ 219090.bmp } & UB & 2810.16 & 2810.16 & 2624.72 & \textbf{ 2501.27 } & 2507.80 & \textbf{ 2501.27 } \\*
      & LB & 2499.72 & 2499.93 & 2379.46 & \textbf{ 2501.27 } & $\dagger$ & $\dagger$ \\*
      & runtime(s) & 0.16 & 0.55 & \textbf{ 0.02 } & 0.14 & 0.48 & 0.97 \\ \hline
   \multirow{3}{*}{ 220075.bmp } & UB & 3265.75 & 3269.41 & 3155.80 & \textbf{ 3115.95 } & 3124.04 & \textbf{ 3115.95 } \\*
      & LB & 3110.83 & 3111.22 & 2939.67 & \textbf{ 3115.95 } & $\dagger$ & $\dagger$ \\*
      & runtime(s) & 0.16 & 1.46 & \textbf{ 0.01 } & 0.38 & 0.75 & 1.65 \\ \hline
   \multirow{3}{*}{ 223061.bmp } & UB & 6812.98 & 6818.95 & 6751.72 & \textbf{ 6576.83 } & 6591.07 & 6578.02 \\*
      & LB & 6565.07 & 6566.44 & 5983.87 & \textbf{ 6576.83 } & $\dagger$ & $\dagger$ \\*
      & runtime(s) & 89.00 & 248.06 & \textbf{ 0.34 } & 37.60 & 1.67 & 1.53 \\ \hline
   \multirow{3}{*}{ 227092.bmp } & UB & 2178.48 & 2180.95 & 2051.77 & \textbf{ 1998.46 } & 2001.58 & 2001.19 \\*
      & LB & 1986.49 & 1986.11 & 1816.59 & \textbf{ 1998.46 } & $\dagger$ & $\dagger$ \\*
      & runtime(s) & 13.93 & 26.75 & \textbf{ 0.03 } & 4.25 & 0.37 & 0.34 \\ \hline
   \multirow{3}{*}{ 229036.bmp } & UB & 6681.13 & 6684.62 & 6250.80 & \textbf{ 6125.73 } & 6153.41 & 6126.61 \\*
      & LB & 6115.98 & 6116.08 & 5729.80 & \textbf{ 6125.73 } & $\dagger$ & $\dagger$ \\*
      & runtime(s) & 4.32 & 12.41 & \textbf{ 0.09 } & 4.83 & 1.53 & 2.00 \\ \hline
   \multirow{3}{*}{ 236037.bmp } & UB & 9324.06 & 9325.93 & 9507.86 & \textbf{ 9060.84 } & 9071.82 & \textbf{ 9060.84 } \\*
      & LB & 9047.32 & 9048.52 & 8261.54 & \textbf{ 9060.84 } & $\dagger$ & $\dagger$ \\*
      & runtime(s) & 62.98 & 166.69 & \textbf{ 0.75 } & 20.14 & 5.02 & 4.38 \\ \hline
   \multirow{3}{*}{ 24077.bmp } & UB & 4846.51 & 4847.62 & 4892.44 & \textbf{ 4761.98 } & 4779.26 & \textbf{ 4761.98 } \\*
      & LB & 4760.74 & 4760.88 & 4531.31 & \textbf{ 4761.98 } & $\dagger$ & $\dagger$ \\*
      & runtime(s) & 3.06 & 5.78 & \textbf{ 0.03 } & 5.10 & 1.14 & 2.40 \\ \hline
   \multirow{3}{*}{ 241004.bmp } & UB & 1271.42 & 1272.99 & 1077.34 & \textbf{ 1057.14 } & 1060.90 & 1057.42 \\*
      & LB & 1056.38 & 1056.34 & 984.58 & \textbf{ 1057.14 } & $\dagger$ & $\dagger$ \\*
      & runtime(s) & 0.18 & 0.92 & \textbf{ 0.00 } & 0.15 & 0.13 & 0.26 \\ \hline
   \multirow{3}{*}{ 241048.bmp } & UB & 4972.81 & 4972.81 & 4945.52 & \textbf{ 4730.95 } & 4750.17 & 4731.19 \\*
      & LB & 4719.03 & 4719.99 & 4343.63 & \textbf{ 4730.95 } & $\dagger$ & $\dagger$ \\*
      & runtime(s) & 4.69 & 26.23 & \textbf{ 0.10 } & 17.26 & 0.79 & 2.00 \\ \hline
   \multirow{3}{*}{ 253027.bmp } & UB & 6839.82 & 6841.19 & 6950.82 & \textbf{ 6606.62 } & 6614.15 & \textbf{ 6606.62 } \\*
      & LB & 6603.97 & 6604.66 & 6371.53 & \textbf{ 6606.62 } & $\dagger$ & $\dagger$ \\*
      & runtime(s) & 18.42 & 49.86 & \textbf{ 0.32 } & 7.39 & 2.90 & 2.36 \\ \hline
   \multirow{3}{*}{ 253055.bmp } & UB & 1684.63 & 1684.63 & 1549.84 & \textbf{ 1502.16 } & 1518.91 & \textbf{ 1502.16 } \\*
      & LB & 1497.70 & 1497.78 & 1408.65 & \textbf{ 1502.16 } & $\dagger$ & $\dagger$ \\*
      & runtime(s) & 0.30 & 0.99 & \textbf{ 0.01 } & 1.31 & 0.09 & 0.19 \\ \hline
   \multirow{3}{*}{ 260058.bmp } & UB & 1458.96 & 1458.96 & 1091.96 & \textbf{ 1084.26 } & 1085.01 & \textbf{ 1084.26 } \\*
      & LB & 1082.23 & 1082.23 & 1017.59 & \textbf{ 1084.26 } & $\dagger$ & $\dagger$ \\*
      & runtime(s) & 0.03 & 0.05 & \textbf{ 0.00 } & 0.18 & 0.10 & 0.15 \\ \hline
   \multirow{3}{*}{ 271035.bmp } & UB & 3706.68 & 3707.59 & 3875.99 & \textbf{ 3621.00 } & 3657.57 & 3621.48 \\*
      & LB & 3613.78 & 3614.44 & 3326.25 & \textbf{ 3621.00 } & $\dagger$ & $\dagger$ \\*
      & runtime(s) & 0.65 & 3.19 & \textbf{ 0.10 } & 12.45 & 0.78 & 1.75 \\ \hline
   \multirow{3}{*}{ 285079.bmp } & UB & 6105.56 & 6105.56 & 5773.41 & \textbf{ 5610.12 } & 5640.22 & 5612.68 \\*
      & LB & 5603.39 & 5603.74 & 5246.17 & \textbf{ 5610.12 } & $\dagger$ & $\dagger$ \\*
      & runtime(s) & 9.30 & 28.78 & \textbf{ 0.10 } & 26.96 & 1.30 & 2.02 \\ \hline
   \multirow{3}{*}{ 291000.bmp } & UB & 10554.63 & 10557.44 & 10401.10 & \textbf{ 10208.87 } & 10222.91 & 10210.47 \\*
      & LB & 10199.14 & 10200.78 & 9626.61 & \textbf{ 10208.87 } & $\dagger$ & $\dagger$ \\*
      & runtime(s) & 13.41 & 44.85 & \textbf{ 0.56 } & 39.32 & 2.46 & 2.24 \\ \hline
   \multirow{3}{*}{ 295087.bmp } & UB & 4562.91 & 4571.66 & 4509.08 & \textbf{ 4290.54 } & 4299.90 & 4291.66 \\*
      & LB & 4288.00 & 4287.83 & 3985.98 & \textbf{ 4290.54 } & $\dagger$ & $\dagger$ \\*
      & runtime(s) & 3.71 & 12.35 & \textbf{ 0.11 } & 6.34 & 1.43 & 1.83 \\ \hline
   \multirow{3}{*}{ 296007.bmp } & UB & 2503.63 & 2507.50 & 2384.77 & \textbf{ 2293.13 } & 2306.24 & 2293.83 \\*
      & LB & 2290.91 & 2290.91 & 2115.21 & \textbf{ 2293.13 } & $\dagger$ & $\dagger$ \\*
      & runtime(s) & 0.32 & 0.63 & \textbf{ 0.03 } & 0.34 & 0.28 & 0.47 \\ \hline
   \multirow{3}{*}{ 296059.bmp } & UB & 2240.64 & 2241.78 & 2160.60 & \textbf{ 2044.71 } & 2045.98 & \textbf{ 2044.71 } \\*
      & LB & 2039.19 & 2039.18 & 1891.55 & \textbf{ 2044.71 } & $\dagger$ & $\dagger$ \\*
      & runtime(s) & 0.57 & 0.99 & \textbf{ 0.02 } & 1.14 & 0.27 & 0.24 \\ \hline
   \multirow{3}{*}{ 299086.bmp } & UB & 1570.66 & 1570.84 & 1650.07 & \textbf{ 1557.24 } & 1561.49 & \textbf{ 1557.24 } \\*
      & LB & 1550.54 & 1550.56 & 1450.72 & \textbf{ 1557.24 } & $\dagger$ & $\dagger$ \\*
      & runtime(s) & 0.47 & 2.10 & \textbf{ 0.02 } & 0.11 & 0.29 & 0.41 \\ \hline
   \multirow{3}{*}{ 300091.bmp } & UB & 1839.69 & 1839.69 & 1508.08 & \textbf{ 1495.10 } & 1498.51 & \textbf{ 1495.10 } \\*
      & LB & 1487.72 & 1487.90 & 1426.88 & \textbf{ 1495.10 } & $\dagger$ & $\dagger$ \\*
      & runtime(s) & 0.16 & 0.49 & \textbf{ 0.01 } & 1.91 & 0.21 & 0.16 \\ \hline
   \multirow{3}{*}{ 302008.bmp } & UB & 2616.47 & 2616.47 & 2583.88 & \textbf{ 2543.23 } & 2557.22 & \textbf{ 2543.23 } \\*
      & LB & 2539.89 & 2539.89 & 2481.98 & \textbf{ 2543.23 } & $\dagger$ & $\dagger$ \\*
      & runtime(s) & 0.21 & 0.54 & \textbf{ 0.01 } & 0.23 & 0.17 & 0.46 \\ \hline
   \multirow{3}{*}{ 304034.bmp } & UB & 8090.64 & 8090.64 & 8159.63 & \textbf{ 7835.47 } & 7850.78 & 7836.49 \\*
      & LB & 7824.86 & 7827.27 & 7212.50 & \textbf{ 7835.47 } & $\dagger$ & $\dagger$ \\*
      & runtime(s) & 9.24 & 54.06 & \textbf{ 0.40 } & 24.88 & 3.14 & 2.05 \\ \hline
   \multirow{3}{*}{ 304074.bmp } & UB & 4490.66 & 4497.86 & 4138.04 & \textbf{ 3891.88 } & 3897.33 & \textbf{ 3891.88 } \\*
      & LB & 3883.34 & 3882.52 & 3543.87 & \textbf{ 3891.88 } & $\dagger$ & $\dagger$ \\*
      & runtime(s) & 3.48 & 6.21 & \textbf{ 0.11 } & 0.51 & 0.83 & 0.79 \\ \hline
   \multirow{3}{*}{ 306005.bmp } & UB & 4683.24 & 4690.02 & 4497.63 & \textbf{ 4290.25 } & 4305.78 & \textbf{ 4290.25 } \\*
      & LB & 4286.33 & 4286.33 & 4004.02 & \textbf{ 4290.25 } & $\dagger$ & $\dagger$ \\*
      & runtime(s) & 5.25 & 14.67 & \textbf{ 0.11 } & 17.38 & 0.91 & 0.98 \\ \hline
   \multirow{3}{*}{ 3096.bmp } & UB & \textbf{ 295.31 } & \textbf{ 295.31 } & 396.90 & 396.90 & 411.23 & 396.90 \\*
      & LB & 388.89 & 388.89 & 389.83 & \textbf{ 396.90 } & $\dagger$ & $\dagger$ \\*
      & runtime(s) & \textbf{ 0.00 } & 0.00 & \textbf{ 0.00 } & 0.02 & 0.00 & 0.06 \\ \hline
   \multirow{3}{*}{ 33039.bmp } & UB & 8286.04 & 8286.04 & 8555.09 & \textbf{ 8069.67 } & 8082.04 & 8070.48 \\*
      & LB & 8061.19 & 8061.78 & 7384.06 & \textbf{ 8069.67 } & $\dagger$ & $\dagger$ \\*
      & runtime(s) & 27.46 & 67.13 & \textbf{ 0.44 } & 20.06 & 3.61 & 3.29 \\ \hline
   \multirow{3}{*}{ 351093.bmp } & UB & 6152.76 & 6155.04 & 6342.60 & \textbf{ 6105.28 } & 6111.45 & 6107.08 \\*
      & LB & 6096.66 & 6097.34 & 5679.10 & \textbf{ 6105.28 } & $\dagger$ & $\dagger$ \\*
      & runtime(s) & 19.48 & 64.42 & \textbf{ 0.20 } & 19.54 & 2.10 & 2.30 \\ \hline
   \multirow{3}{*}{ 361010.bmp } & UB & 3500.13 & 3503.88 & 3459.05 & \textbf{ 3361.02 } & 3368.12 & 3364.98 \\*
      & LB & 3356.03 & 3356.01 & 3189.41 & \textbf{ 3361.02 } & $\dagger$ & $\dagger$ \\*
      & runtime(s) & 0.42 & 1.06 & \textbf{ 0.02 } & 0.39 & 0.91 & 0.79 \\ \hline
   \multirow{3}{*}{ 37073.bmp } & UB & 1991.36 & 1991.36 & 2044.68 & \textbf{ 1975.00 } & 1982.24 & \textbf{ 1975.00 } \\*
      & LB & 1972.11 & 1972.11 & 1904.57 & \textbf{ 1975.00 } & $\dagger$ & $\dagger$ \\*
      & runtime(s) & 0.14 & 0.30 & \textbf{ 0.01 } & 0.15 & 0.31 & 0.49 \\ \hline
   \multirow{3}{*}{ 376043.bmp } & UB & 6549.11 & 6557.74 & 6054.26 & \textbf{ 5863.83 } & 5872.57 & 5864.00 \\*
      & LB & 5859.62 & 5860.01 & 5433.92 & \textbf{ 5863.83 } & $\dagger$ & $\dagger$ \\*
      & runtime(s) & 7.34 & 31.66 & \textbf{ 0.17 } & 9.80 & 1.32 & 0.94 \\ \hline
   \multirow{3}{*}{ 38082.bmp } & UB & 8324.62 & 8324.62 & 8492.54 & \textbf{ 8060.34 } & 8065.57 & 8066.62 \\*
      & LB & 8047.84 & 8048.69 & 7359.93 & \textbf{ 8060.34 } & $\dagger$ & $\dagger$ \\*
      & runtime(s) & 439.50 & 702.48 & \textbf{ 0.62 } & 26.23 & 2.41 & 2.49 \\ \hline
   \multirow{3}{*}{ 38092.bmp } & UB & 4501.49 & 4511.20 & 4213.16 & \textbf{ 4071.86 } & 4085.14 & \textbf{ 4071.86 } \\*
      & LB & 4067.32 & 4067.13 & 3814.35 & \textbf{ 4071.86 } & $\dagger$ & $\dagger$ \\*
      & runtime(s) & 1.24 & 2.30 & \textbf{ 0.09 } & 0.46 & 0.91 & 0.99 \\ \hline
   \multirow{3}{*}{ 385039.bmp } & UB & 3995.35 & 3995.35 & 3876.12 & \textbf{ 3745.53 } & 3752.81 & \textbf{ 3745.53 } \\*
      & LB & 3741.30 & 3742.39 & 3565.10 & \textbf{ 3745.53 } & $\dagger$ & $\dagger$ \\*
      & runtime(s) & 0.42 & 3.10 & \textbf{ 0.06 } & 2.51 & 1.20 & 0.81 \\ \hline
   \multirow{3}{*}{ 41033.bmp } & UB & 2585.71 & 2585.71 & 2050.50 & \textbf{ 1994.24 } & 2001.06 & 1997.55 \\*
      & LB & 1988.64 & 1989.20 & 1841.58 & \textbf{ 1994.24 } & $\dagger$ & $\dagger$ \\*
      & runtime(s) & 0.54 & 4.62 & \textbf{ 0.01 } & 0.94 & 0.28 & 0.40 \\ \hline
   \multirow{3}{*}{ 41069.bmp } & UB & 6685.12 & 6685.12 & 5182.46 & \textbf{ 5110.96 } & 5115.12 & 5122.63 \\*
      & LB & 5091.16 & 5092.47 & 4896.23 & \textbf{ 5110.96 } & $\dagger$ & $\dagger$ \\*
      & runtime(s) & 24.57 & 59.03 & \textbf{ 0.07 } & 35.77 & 0.36 & 0.54 \\ \hline
   \multirow{3}{*}{ 42012.bmp } & UB & 3561.54 & 3561.54 & 3524.18 & \textbf{ 3248.70 } & 3252.28 & 3251.04 \\*
      & LB & 3238.97 & 3240.27 & 3005.31 & \textbf{ 3248.70 } & $\dagger$ & $\dagger$ \\*
      & runtime(s) & 3.96 & 27.65 & \textbf{ 0.08 } & 8.95 & 0.59 & 0.73 \\ \hline
   \multirow{3}{*}{ 42049.bmp } & UB & \textbf{ 970.16 } & \textbf{ 970.16 } & 1098.96 & 1069.22 & 1076.54 & 1069.22 \\*
      & LB & 996.85 & 996.85 & 997.53 & \textbf{ 1069.22 } & $\dagger$ & $\dagger$ \\*
      & runtime(s) & \textbf{ 0.00 } & \textbf{ 0.00 } & \textbf{ 0.00 } & 0.22 & 0.10 & 0.31 \\ \hline
   \multirow{3}{*}{ 43074.bmp } & UB & 2864.46 & 2871.95 & 2374.97 & \textbf{ 2332.83 } & 2340.46 & 2333.36 \\*
      & LB & 2329.22 & 2329.38 & 2166.88 & \textbf{ 2332.83 } & $\dagger$ & $\dagger$ \\*
      & runtime(s) & 1.82 & 6.38 & \textbf{ 0.03 } & 5.82 & 0.32 & 0.28 \\ \hline
   \multirow{3}{*}{ 45096.bmp } & UB & 1088.06 & 1092.36 & 1024.96 & \textbf{ 977.78 } & 1031.06 & \textbf{ 977.78 } \\*
      & LB & 975.97 & 975.97 & 911.19 & \textbf{ 977.78 } & $\dagger$ & $\dagger$ \\*
      & runtime(s) & 0.06 & 0.12 & \textbf{ 0.01 } & 0.10 & 0.05 & 0.25 \\ \hline
   \multirow{3}{*}{ 54082.bmp } & UB & 4075.17 & 4075.17 & 3914.51 & \textbf{ 3796.36 } & 3806.52 & 3797.18 \\*
      & LB & 3785.04 & 3785.94 & 3486.46 & \textbf{ 3796.36 } & $\dagger$ & $\dagger$ \\*
      & runtime(s) & 4.80 & 36.47 & \textbf{ 0.10 } & 5.80 & 0.43 & 1.03 \\ \hline
   \multirow{3}{*}{ 55073.bmp } & UB & 8445.42 & 8449.44 & 8179.32 & \textbf{ 7835.96 } & 7844.37 & 7838.99 \\*
      & LB & 7820.19 & 7822.05 & 7189.03 & \textbf{ 7835.96 } & $\dagger$ & $\dagger$ \\*
      & runtime(s) & 54.41 & 159.67 & \textbf{ 0.39 } & 13.66 & 2.54 & 1.67 \\ \hline
   \multirow{3}{*}{ 58060.bmp } & UB & 10196.99 & 10197.09 & 10133.69 & \textbf{ 9881.86 } & 9891.88 & 9882.77 \\*
      & LB & 9877.23 & 9878.45 & 9389.50 & \textbf{ 9881.86 } & $\dagger$ & $\dagger$ \\*
      & runtime(s) & 35.63 & 155.56 & \textbf{ 0.27 } & 20.02 & 4.15 & 6.54 \\ \hline
   \multirow{3}{*}{ 62096.bmp } & UB & 3915.47 & 3915.47 & 3485.31 & \textbf{ 3419.40 } & 3421.38 & 3420.03 \\*
      & LB & 3413.39 & 3412.83 & 3233.59 & \textbf{ 3419.40 } & $\dagger$ & $\dagger$ \\*
      & runtime(s) & 1.37 & 5.87 & \textbf{ 0.05 } & 5.70 & 0.54 & 0.56 \\ \hline
   \multirow{3}{*}{ 65033.bmp } & UB & 7766.31 & 7767.46 & 7647.20 & \textbf{ 7364.57 } & 7372.11 & 7365.30 \\*
      & LB & 7360.45 & 7360.57 & 6865.71 & \textbf{ 7364.57 } & $\dagger$ & $\dagger$ \\*
      & runtime(s) & 5.16 & 15.77 & \textbf{ 0.18 } & 11.87 & 2.04 & 2.05 \\ \hline
   \multirow{3}{*}{ 66053.bmp } & UB & 4871.59 & 4871.59 & 4522.25 & \textbf{ 4427.25 } & 4434.10 & \textbf{ 4427.25 } \\*
      & LB & 4417.69 & 4418.10 & 4172.50 & \textbf{ 4427.25 } & $\dagger$ & $\dagger$ \\*
      & runtime(s) & 3.11 & 7.20 & \textbf{ 0.07 } & 4.29 & 0.65 & 0.80 \\ \hline
   \multirow{3}{*}{ 69015.bmp } & UB & 4378.98 & 4378.98 & 4183.41 & \textbf{ 4024.45 } & 4032.54 & 4025.34 \\*
      & LB & 4019.16 & 4019.52 & 3778.35 & \textbf{ 4024.45 } & $\dagger$ & $\dagger$ \\*
      & runtime(s) & 1.66 & 4.08 & \textbf{ 0.06 } & 6.14 & 1.00 & 1.27 \\ \hline
   \multirow{3}{*}{ 69020.bmp } & UB & 5828.19 & 5831.93 & 5527.82 & \textbf{ 5179.29 } & 5183.07 & \textbf{ 5179.29 } \\*
      & LB & 5170.40 & 5170.63 & 4796.09 & \textbf{ 5179.29 } & $\dagger$ & $\dagger$ \\*
      & runtime(s) & 45.51 & 84.91 & \textbf{ 0.20 } & 12.34 & 1.16 & 0.97 \\ \hline
   \multirow{3}{*}{ 69040.bmp } & UB & 8240.10 & 8242.44 & 8255.62 & \textbf{ 7974.58 } & 7994.72 & 7983.83 \\*
      & LB & 7955.22 & 7958.22 & 7213.49 & \textbf{ 7974.58 } & $\dagger$ & $\dagger$ \\*
      & runtime(s) & 230.92 & 371.04 & \textbf{ 0.55 } & 31.49 & 2.16 & 3.76 \\ \hline
   \multirow{3}{*}{ 76053.bmp } & UB & 4625.92 & 4625.92 & 4823.04 & \textbf{ 4514.99 } & 4527.26 & 4516.03 \\*
      & LB & 4504.16 & 4505.47 & 4073.68 & \textbf{ 4514.99 } & $\dagger$ & $\dagger$ \\*
      & runtime(s) & 13.42 & 35.32 & \textbf{ 0.14 } & 18.85 & 1.59 & 1.21 \\ \hline
   \multirow{3}{*}{ 78004.bmp } & UB & 3394.42 & 3394.80 & 3380.61 & \textbf{ 3254.61 } & 3271.58 & 3254.85 \\*
      & LB & 3248.43 & 3248.75 & 3109.85 & \textbf{ 3254.61 } & $\dagger$ & $\dagger$ \\*
      & runtime(s) & 0.41 & 1.67 & \textbf{ 0.01 } & 2.70 & 0.34 & 0.64 \\ \hline
   \multirow{3}{*}{ 8023.bmp } & UB & 4385.74 & 4388.92 & 4108.08 & \textbf{ 4023.38 } & 4032.92 & 4026.67 \\*
      & LB & 4019.49 & 4019.76 & 3679.77 & \textbf{ 4023.38 } & $\dagger$ & $\dagger$ \\*
      & runtime(s) & 16.96 & 56.59 & \textbf{ 0.14 } & 32.67 & 0.84 & 0.60 \\ \hline
   \multirow{3}{*}{ 85048.bmp } & UB & 6056.07 & 6056.07 & 6186.51 & \textbf{ 5851.38 } & 5863.69 & 5852.33 \\*
      & LB & 5844.50 & 5845.08 & 5452.57 & \textbf{ 5851.38 } & $\dagger$ & $\dagger$ \\*
      & runtime(s) & 6.40 & 18.79 & \textbf{ 0.16 } & 8.07 & 2.03 & 1.91 \\ \hline
   \multirow{3}{*}{ 86000.bmp } & UB & \textbf{ 4628.18 } & \textbf{ 4628.18 } & 4769.21 & 4633.86 & 4643.13 & 4633.96 \\*
      & LB & 4628.53 & 4628.53 & 4414.10 & \textbf{ 4633.86 } & $\dagger$ & $\dagger$ \\*
      & runtime(s) & 0.31 & 0.33 & \textbf{ 0.03 } & 4.90 & 1.49 & 1.76 \\ \hline
   \multirow{3}{*}{ 86016.bmp } & UB & 7930.26 & 7930.26 & 6654.84 & \textbf{ 6618.85 } & 6619.75 & 6620.36 \\*
      & LB & 6617.36 & 6617.46 & 6502.79 & \textbf{ 6618.85 } & $\dagger$ & $\dagger$ \\*
      & runtime(s) & 1.96 & 5.10 & \textbf{ 0.02 } & 2.62 & 0.36 & 0.53 \\ \hline
   \multirow{3}{*}{ 86068.bmp } & UB & 5870.65 & 5876.08 & 5289.20 & \textbf{ 5198.87 } & 5207.65 & 5205.68 \\*
      & LB & 5186.42 & 5185.40 & 4731.74 & \textbf{ 5198.87 } & $\dagger$ & $\dagger$ \\*
      & runtime(s) & 28.82 & 75.20 & \textbf{ 0.23 } & 9.72 & 1.10 & 0.60 \\ \hline
   \multirow{3}{*}{ 87046.bmp } & UB & 4641.36 & 4641.36 & 4470.22 & \textbf{ 4315.53 } & 4321.55 & \textbf{ 4315.53 } \\*
      & LB & 4304.42 & 4305.38 & 3985.82 & \textbf{ 4315.53 } & $\dagger$ & $\dagger$ \\*
      & runtime(s) & 4.69 & 9.11 & \textbf{ 0.09 } & 11.52 & 1.14 & 0.60 \\ \hline
   \multirow{3}{*}{ 89072.bmp } & UB & 4096.15 & 4098.40 & 4159.28 & \textbf{ 3933.75 } & 3948.57 & 3934.47 \\*
      & LB & 3925.22 & 3925.52 & 3707.83 & \textbf{ 3933.75 } & $\dagger$ & $\dagger$ \\*
      & runtime(s) & 0.93 & 3.41 & \textbf{ 0.06 } & 3.69 & 1.07 & 1.00 \\ \hline
   \multirow{3}{*}{ 97033.bmp } & UB & 4594.50 & 4594.50 & 4583.48 & \textbf{ 4320.69 } & 4336.56 & 4322.76 \\*
      & LB & 4311.08 & 4312.65 & 3996.62 & \textbf{ 4320.69 } & $\dagger$ & $\dagger$ \\*
      & runtime(s) & 2.36 & 9.46 & \textbf{ 0.06 } & 1.62 & 0.83 & 1.31 \\ \hline
\multicolumn{ 8 }{| c |}{ \textbf{ modularity clustering } } \\* \hline 
   \multirow{3}{*}{ adjnoun } & UB & \textbf{ -0.31 } & \textbf{ -0.31 } & -0.17 & 0.00 & 0.00 & -0.29 \\*
      & LB & -0.46 & -0.46 & -0.79 & \textbf{ -0.45 } & $\dagger$ & $\dagger$ \\*
      & runtime(s) & 1.10 & 1.11 & 0.29 & 7288.07 & \textbf{ 0.01 } & 92.48 \\ \hline
   \multirow{3}{*}{ dolphins } & UB & -0.53 & -0.53 & -0.34 & \textbf{ -0.53 } & 0.00 & -0.52 \\*
      & LB & -0.55 & -0.55 & -0.83 & \textbf{ -0.53 } & $\dagger$ & $\dagger$ \\*
      & runtime(s) & 0.24 & 0.22 & 0.03 & 44.61 & \textbf{ 0.00 } & 0.79 \\ \hline
   \multirow{3}{*}{ football } & UB & \textbf{ -0.60 } & \textbf{ -0.60 } & -0.34 & -0.60 & 0.00 & -0.49 \\*
      & LB & -0.62 & -0.62 & -0.90 & \textbf{ -0.60 } & $\dagger$ & $\dagger$ \\*
      & runtime(s) & 1.99 & 1.98 & 0.37 & 71.91 & \textbf{ 0.01 } & 7.63 \\ \hline
   \multirow{3}{*}{ karate } & UB & \textbf{ -0.42 } & \textbf{ -0.42 } & -0.28 & -0.42 & 0.00 & -0.32 \\*
      & LB & -0.43 & -0.43 & -0.66 & \textbf{ -0.42 } & $\dagger$ & $\dagger$ \\*
      & runtime(s) & 0.03 & 0.06 & 0.00 & 0.40 & \textbf{ 0.00 } & 0.08 \\ \hline
   \multirow{3}{*}{ lesmis } & UB & -0.56 & -0.56 & -0.37 & \textbf{ -0.56 } & 0.00 & -0.50 \\*
      & LB & -0.57 & -0.57 & -0.72 & \textbf{ -0.56 } & $\dagger$ & $\dagger$ \\*
      & runtime(s) & 0.18 & 0.83 & 0.03 & 3.77 & \textbf{ 0.00 } & 0.72 \\ \hline
   \multirow{3}{*}{ polbooks } & UB & -0.52 & -0.52 & -0.33 & \textbf{ -0.53 } & 0.00 & -0.51 \\*
      & LB & -0.56 & -0.56 & -0.83 & \textbf{ -0.54 } & $\dagger$ & $\dagger$ \\*
      & runtime(s) & 0.86 & 0.92 & 0.17 & 10057.85 & \textbf{ 0.01 } & 4.97 \\ \hline
\multicolumn{ 8 }{| c |}{ \textbf{ knott-3d-150 } } \\* \hline 
   \multirow{3}{*}{ gm\_knott\_3d\_032 } & UB & \textbf{ -5811.47 } & \textbf{ -5811.47 } & -5365.34 & -5811.47 & -5745.79 & -5767.34 \\*
      & LB & -5812.64 & -5811.49 & -6052.81 & \textbf{ -5811.47 } & $\dagger$ & $\dagger$ \\*
      & runtime(s) & 0.47 & 2.68 & \textbf{ 0.03 } & 2.70 & 0.17 & 0.48 \\ \hline
   \multirow{3}{*}{ gm\_knott\_3d\_033 } & UB & \textbf{ -2545.84 } & \textbf{ -2545.84 } & -2536.26 & -2545.84 & -2517.65 & -2545.84 \\*
      & LB & -2545.90 & -2545.90 & -3029.52 & \textbf{ -2545.84 } & $\dagger$ & $\dagger$ \\*
      & runtime(s) & 0.41 & 0.88 & \textbf{ 0.02 } & 1.28 & 0.22 & 0.35 \\ \hline
   \multirow{3}{*}{ gm\_knott\_3d\_034 } & UB & \textbf{ -4064.87 } & \textbf{ -4064.87 } & -3921.65 & -4064.87 & -3971.60 & -3972.66 \\*
      & LB & -4066.65 & -4066.31 & -4337.06 & \textbf{ -4064.87 } & $\dagger$ & $\dagger$ \\*
      & runtime(s) & 0.51 & 1.26 & \textbf{ 0.02 } & 3.69 & 0.47 & 0.58 \\ \hline
   \multirow{3}{*}{ gm\_knott\_3d\_035 } & UB & \textbf{ -4595.84 } & \textbf{ -4595.84 } & -4238.86 & -4595.84 & -4568.88 & -4595.84 \\*
      & LB & -4595.84 & -4595.84 & -4916.47 & \textbf{ -4595.84 } & $\dagger$ & $\dagger$ \\*
      & runtime(s) & 0.36 & 0.49 & \textbf{ 0.08 } & 1.18 & 0.27 & 0.46 \\ \hline
   \multirow{3}{*}{ gm\_knott\_3d\_036 } & UB & -5192.26 & -5192.26 & -4678.23 & \textbf{ -5198.37 } & -5159.18 & \textbf{ -5198.37 } \\*
      & LB & -5199.32 & -5198.80 & -5440.26 & \textbf{ -5198.37 } & $\dagger$ & $\dagger$ \\*
      & runtime(s) & 0.35 & 1.01 & \textbf{ 0.09 } & 1.32 & 0.32 & 0.62 \\ \hline
   \multirow{3}{*}{ gm\_knott\_3d\_037 } & UB & -4636.50 & -4633.67 & -4312.24 & \textbf{ -4638.99 } & -4616.39 & -4638.03 \\*
      & LB & -4642.64 & -4639.28 & -4928.85 & \textbf{ -4638.99 } & $\dagger$ & $\dagger$ \\*
      & runtime(s) & 3.53 & 18.53 & \textbf{ 0.03 } & 6.27 & 0.23 & 0.76 \\ \hline
   \multirow{3}{*}{ gm\_knott\_3d\_038 } & UB & \textbf{ -4625.80 } & \textbf{ -4625.80 } & -4235.69 & -4625.80 & -4616.99 & -4619.01 \\*
      & LB & -4625.80 & -4625.80 & -4818.46 & \textbf{ -4625.80 } & $\dagger$ & $\dagger$ \\*
      & runtime(s) & 0.35 & 1.52 & \textbf{ 0.05 } & 0.65 & 0.19 & 0.46 \\ \hline
   \multirow{3}{*}{ gm\_knott\_3d\_039 } & UB & -5092.32 & -5092.32 & -4476.96 & \textbf{ -5092.32 } & -5081.58 & -5082.97 \\*
      & LB & -5092.45 & \textbf{ -5092.32 } & -5318.04 & -5092.32 & $\dagger$ & $\dagger$ \\*
      & runtime(s) & 0.73 & 4.11 & \textbf{ 0.03 } & 1.88 & 0.23 & 0.58 \\ \hline
\multicolumn{ 8 }{| c |}{ \textbf{ knott-3d-300 } } \\* \hline 
  \multirow{3}{*}{ gm\_knott\_3d\_072 } & UB & -32986.39 & -32986.39 & -29632.23 & \textbf{ -32999.85 } & -32883.17 & -32875.45 \\*
      & LB & -33006.26 & -33006.26 & -34512.56 & \textbf{ -32999.85 } & $\dagger$ & $\dagger$ \\*
      & runtime(s) & 1466.31 & 1574.69 & 4.60 & 49.78 & 4.26 & \textbf{ 3.40 } \\ \hline
   \multirow{3}{*}{ gm\_knott\_3d\_073 } & UB & -25863.15 & -25863.15 & -23433.69 & \textbf{ -25863.38 } & -25738.25 & -25740.05 \\*
      & LB & -25866.11 & -25863.69 & -27464.92 & \textbf{ -25863.38 } & $\dagger$ & $\dagger$ \\*
      & runtime(s) & 58.50 & 81.89 & 3.24 & 55.72 & \textbf{ 2.72 } & 8.88 \\ \hline
   \multirow{3}{*}{ gm\_knott\_3d\_074 } & UB & -25685.56 & -25685.03 & -23513.94 & \textbf{ -25721.90 } & -25625.74 & -25627.65 \\*
      & LB & -25726.98 & -25723.03 & -27196.88 & \textbf{ -25721.90 } & $\dagger$ & $\dagger$ \\*
      & runtime(s) & 172.42 & 229.99 & \textbf{ 1.19 } & 39.63 & 2.18 & 10.73 \\ \hline
   \multirow{3}{*}{ gm\_knott\_3d\_075 } & UB & -30456.95 & -30455.38 & -27294.35 & \textbf{ -30478.37 } & -30429.06 & -30471.30 \\*
      & LB & -30480.69 & -30478.37 & -31854.88 & \textbf{ -30478.37 } & $\dagger$ & $\dagger$ \\*
      & runtime(s) & 26.07 & 46.55 & 2.89 & 14.29 & \textbf{ 2.18 } & 3.57 \\ \hline
   \multirow{3}{*}{ gm\_knott\_3d\_076 } & UB & -27000.93 & -27000.93 & -24789.93 & \textbf{ -27056.99 } & -27004.21 & -27031.94 \\*
      & LB & -27065.92 & -27060.80 & -28550.56 & \textbf{ -27056.99 } & $\dagger$ & $\dagger$ \\*
      & runtime(s) & 1456.89 & 2236.72 & \textbf{ 2.62 } & 50.33 & 3.32 & 10.79 \\ \hline
   \multirow{3}{*}{ gm\_knott\_3d\_077 } & UB & \textbf{ -29482.24 } & \textbf{ -29482.24 } & -27122.85 & -29482.24 & -29476.76 & -29481.33 \\*
      & LB & -29482.55 & -29482.26 & -31159.83 & \textbf{ -29482.24 } & $\dagger$ & $\dagger$ \\*
      & runtime(s) & 170.01 & 155.43 & \textbf{ 2.40 } & 47.23 & 4.12 & 10.59 \\ \hline
   \multirow{3}{*}{ gm\_knott\_3d\_078 } & UB & -20206.78 & -20206.78 & -19374.02 & \textbf{ -20211.55 } & -20189.82 & -20157.27 \\*
      & LB & -20217.62 & -20214.30 & -22015.47 & \textbf{ -20211.55 } & $\dagger$ & $\dagger$ \\*
      & runtime(s) & 350.20 & 1193.95 & \textbf{ 1.56 } & 1451.04 & 2.44 & 10.97 \\ \hline
   \multirow{3}{*}{ gm\_knott\_3d\_079 } & UB & -26601.32 & -26601.32 & -23755.71 & \textbf{ -26607.98 } & -26589.21 & -26593.33 \\*
      & LB & -26612.76 & -26608.39 & -28457.54 & \textbf{ -26607.98 } & $\dagger$ & $\dagger$ \\*
      & runtime(s) & 104.62 & 463.29 & 3.37 & 110.62 & \textbf{ 2.46 } & 6.24 \\ \hline
\multicolumn{ 8 }{| c |}{ \textbf{ knott-3d-450 } } \\* \hline 
      \multirow{3}{*}{ gm\_knott\_3d\_096 } & UB & -89941.58 & -89941.58 & -80280.68 & \textbf{ -89959.41 } & -89779.76 & -89786.75 \\*
      & LB & -89996.73 & -89996.38 & -94633.08 & \textbf{ -89959.41 } & $\dagger$ & $\dagger$ \\*
      & runtime(s) & 3643.00 & 3614.63 & 48.66 & 1438.60 & \textbf{ 16.89 } & 153.09 \\ \hline 
      \multirow{3}{*}{ gm\_knott\_3d\_097 } & UB & -73473.47 & -73473.47 & -67333.36 & \textbf{ -73477.55 } & -73386.50 & -73364.35 \\*
      & LB & -73516.42 & -73516.60 & -78026.62 & \textbf{ -73477.55 } & $\dagger$ & $\dagger$ \\*
      & runtime(s) & 3678.62 & 3674.28 & 26.83 & 1075.22 & \textbf{ 15.59 } & 177.63 \\ \hline 
      \multirow{3}{*}{ gm\_knott\_3d\_098 } & UB & -86499.41 & -86499.41 & -78396.23 & \textbf{ -86593.97 } & -86470.44 & -86498.10 \\*
      & LB & -86633.78 & -86632.24 & -91051.86 & \textbf{ -86593.97 } & $\dagger$ & $\dagger$ \\*
      & runtime(s) & 3648.34 & 3623.29 & \textbf{ 11.39 } & 1507.50 & 15.15 & 69.82 \\ \hline
   \multirow{3}{*}{ gm\_knott\_3d\_099 } & UB & -86177.37 & -86177.37 & -78712.34 & -85956.93 & \textbf{ -86184.26 } & -86180.01 \\*
      & LB & -86320.89 & \textbf{ -86309.52 } & -91020.66 & -86449.40 & $\dagger$ & $\dagger$ \\*
      & runtime(s) & 3634.38 & 3639.38 & 34.96 & 3732.50 & \textbf{ 11.12 } & 90.11 \\ \hline
   \multirow{3}{*}{ gm\_knott\_3d\_100 } & UB & -76590.45 & -76590.45 & -68324.61 & \textbf{ -76699.37 } & -76561.72 & -76523.90 \\*
      & LB & -76761.06 & -76758.81 & -81763.95 & \textbf{ -76699.37 } & $\dagger$ & $\dagger$ \\*
      & runtime(s) & 3628.67 & 3668.23 & 40.53 & 1076.28 & \textbf{ 22.65 } & 141.63 \\ \hline
   \multirow{3}{*}{ gm\_knott\_3d\_101 } & UB & -74508.29 & -74508.29 & -66796.31 & \textbf{ -74529.51 } & -74500.99 & -74495.74 \\*
      & LB & -74544.89 & -74544.24 & -79463.73 & \textbf{ -74529.51 } & $\dagger$ & $\dagger$ \\*
      & runtime(s) & 3616.45 & 3620.81 & 33.62 & 1149.05 & \textbf{ 19.52 } & 110.02 \\ \hline
   \multirow{3}{*}{ gm\_knott\_3d\_102 } & UB & -66423.87 & -66423.87 & -60651.60 & \textbf{ -66482.68 } & -66454.83 & -66455.86 \\*
      & LB & -66525.90 & -66524.32 & -71160.72 & \textbf{ -66482.68 } & $\dagger$ & $\dagger$ \\*
      & runtime(s) & 3680.72 & 3617.03 & 19.75 & 907.64 & \textbf{ 14.44 } & 131.82 \\ \hline
   \multirow{3}{*}{ gm\_knott\_3d\_103 } & UB & \textbf{ -73799.17 } & \textbf{ -73799.17 } & -66427.02 & -73431.17 & -73750.61 & -73743.79 \\*
      & LB & -73918.13 & \textbf{ -73909.00 } & -79062.14 & -73988.21 & $\dagger$ & $\dagger$ \\*
      & runtime(s) & 3662.18 & 3667.64 & 36.71 & 3836.96 & \textbf{ 16.82 } & 79.70 \\ \hline
\multicolumn{ 8 }{| c |}{ \textbf{ knott-3d-550 } } \\* \hline 
  \multirow{3}{*}{ gm\_knott\_3d\_112 } & UB & -152854.11 & -152854.11 & -136448.39 & \textbf{ -153021.45 } & -152908.44 & -152675.73 \\*
      & LB & -153124.66 & -153124.66 & -160981.94 & \textbf{ -153024.89 } & $\dagger$ & $\dagger$ \\*
      & runtime(s) & 3686.25 & 3692.63 & 151.77 & 3716.74 & \textbf{ 88.55 } & 462.34 \\ \hline
   \multirow{3}{*}{ gm\_knott\_3d\_113 } & UB & \textbf{ -135567.04 } & \textbf{ -135567.04 } & -122181.33 & -134820.65 & -135466.46 & -135386.34 \\*
      & LB & \textbf{ -135824.78 } & \textbf{ -135824.78 } & -144181.42 & -135924.12 & $\dagger$ & $\dagger$ \\*
      & runtime(s) & 3657.02 & 3650.21 & 105.96 & 3614.44 & \textbf{ 57.52 } & 540.01 \\ \hline
   \multirow{3}{*}{ gm\_knott\_3d\_114 } & UB & -149545.82 & -149545.82 & -134889.26 & \textbf{ -149716.68 } & -149683.60 & -149526.15 \\*
      & LB & -149831.75 & -149831.31 & -157228.92 & \textbf{ -149722.18 } & $\dagger$ & $\dagger$ \\*
      & runtime(s) & 3651.93 & 3627.15 & 108.75 & 3614.56 & \textbf{ 99.81 } & 384.54 \\ \hline
   \multirow{3}{*}{ gm\_knott\_3d\_115 } & UB & -149769.24 & -149769.24 & -135760.55 & -148726.82 & -149736.65 & \textbf{ -149777.01 } \\*
      & LB & -150320.76 & \textbf{ -150285.41 } & -158348.78 & -150325.94 & $\dagger$ & $\dagger$ \\*
      & runtime(s) & 3647.68 & 3746.48 & 111.32 & 3747.40 & \textbf{ 47.97 } & 384.89 \\ \hline
   \multirow{3}{*}{ gm\_knott\_3d\_116 } & UB & -130577.35 & -130577.35 & -118822.35 & \textbf{ -130757.57 } & -130720.10 & -130580.01 \\*
      & LB & -130910.50 & -130894.84 & -138934.59 & \textbf{ -130761.25 } & $\dagger$ & $\dagger$ \\*
      & runtime(s) & 4152.93 & 3854.59 & \textbf{ 88.56 } & 3688.88 & 92.77 & 976.88 \\ \hline
   \multirow{3}{*}{ gm\_knott\_3d\_117 } & UB & -123419.07 & -123419.07 & -112948.58 & -122646.08 & -123368.23 & \textbf{ -123448.71 } \\*
      & LB & -123849.34 & -123819.48 & -131937.12 & \textbf{ -123810.61 } & $\dagger$ & $\dagger$ \\*
      & runtime(s) & 3860.97 & 3778.57 & 90.49 & 3617.71 & \textbf{ 54.49 } & 582.92 \\ \hline
   \multirow{3}{*}{ gm\_knott\_3d\_118 } & UB & -123467.39 & -123467.39 & -112812.30 & -122526.33 & \textbf{ -123520.61 } & -123483.58 \\*
      & LB & -123720.59 & -123709.54 & -131313.86 & \textbf{ -123538.09 } & $\dagger$ & $\dagger$ \\*
      & runtime(s) & 3720.24 & 3824.32 & \textbf{ 77.63 } & 3777.86 & 83.73 & 869.68 \\ \hline
   \multirow{3}{*}{ gm\_knott\_3d\_119 } & UB & \textbf{ -126318.25 } & \textbf{ -126318.25 } & -116868.98 & -123919.61 & -126308.30 & -126289.59 \\*
      & LB & -126936.67 & -126936.67 & -134702.52 & \textbf{ -126935.77 } & $\dagger$ & $\dagger$ \\*
      & runtime(s) & 3889.64 & 3792.14 & 84.92 & 3688.17 & \textbf{ 58.68 } & 555.56 \\ \hline \pagebreak
\multicolumn{ 8 }{| c |}{ \textbf{ CREMI-small } } \\* \hline
   \multirow{3}{*}{ gm\_small\_1 } & UB & -301663.25 & -301663.25 & -278423.74 & \textbf{ -301674.02 } & \textbf{ 0.00 } & \textbf{ -301674.02 } \\*
      & LB & -301678.18 & -301677.98 & -302379.69 & \textbf{ -301673.92 } & $\dagger$ & $\dagger$ \\*
      & runtime(s) & 3638.89 & 3686.59 & \textbf{ 56.18 } & 998.95 & \textbf{ 0.00 } & 191.08 \\ \hline
   \multirow{3}{*}{ gm\_small\_2 } & UB & \textbf{ -127448.90 } & \textbf{ -127448.90 } & -114182.56 & -116678.57 & -127414.03 & -127292.54 \\*
      & LB & -127545.97 & -127545.60 & -131075.52 & \textbf{ -127520.30 } & $\dagger$ & $\dagger$ \\*
      & runtime(s) & 3623.23 & 3679.25 & \textbf{ 369.68 } & 3722.33 & 3607.92 & 3752.86 \\ \hline
   \multirow{3}{*}{ gm\_small\_3 } & UB & -210390.21 & -210390.21 & -191243.51 & \textbf{ -210430.88 } & -210396.31 & -210386.96 \\*
      & LB & -210452.80 & -210453.44 & -212966.73 & \textbf{ -210432.60 } & $\dagger$ & $\dagger$ \\*
      & runtime(s) & 3674.40 & 3619.29 & \textbf{ 531.16 } & 3606.16 & 3479.31 & 3722.51 \\ \hline
\multicolumn{ 8 }{| c |}{ \textbf{ CREMI-large } } \\* \hline 
   \multirow{3}{*}{ gm\_large\_1 } & UB & \textbf{ -5647807.76 } & \textbf{ -5647807.76 } & \textbf{ 0.00 } & \textbf{ 0.00 } & -5628646.55 & -5524856.64 \\*
      & LB & \textbf{ -5648791.54 } & \textbf{ -5648791.54 } & $\dagger$ & $\dagger$ & 0.00 & 0.00 \\*
      & runtime(s) & 3672.44 & \textbf{ 3629.77 } & \textbf{ 0.00 } & \textbf{ 0.00 } & 3863.35 & 10197.66 \\ \hline
   \multirow{3}{*}{ gm\_large\_2 } & UB & \textbf{ -2368830.08 } & \textbf{ -2368830.08 } & \textbf{ 0.00 } & \textbf{ 0.00 } & -1916548.20 & -1713523.77 \\*
      & LB & \textbf{ -2382103.57 } & \textbf{ -2382103.57 } & $\dagger$ & $\dagger$ & 0.00 & 0.00 \\*
      & runtime(s) & \textbf{ 3613.51 } & 4076.68 & \textbf{ 0.00 } & \textbf{ 0.00 } & 8092.81 & 36081.14 \\ \hline
   \multirow{3}{*}{ gm\_large\_3 } & UB & \textbf{ -3643885.09 } & \textbf{ -3643885.09 } & \textbf{ 0.00 } & \textbf{ 0.00 } & \textbf{ 0.00 } & \textbf{ 0.00 } \\*
      & LB & \textbf{ -3648930.86 } & \textbf{ -3648930.86 } & $\dagger$ & $\dagger$ & 0.00 & 0.00 \\*
      & runtime(s) & 3717.06 & \textbf{ 3712.53 } & \textbf{ 0.00 } & \textbf{ 0.00 } & \textbf{ 0.00 } & \textbf{ 0.00 } \\ \hline
\end{longtable}
} 

\end{document}